\documentclass[11pt]{article}

\usepackage{fullpage}
\usepackage[font={small,it}]{caption}
\usepackage{subcaption}
\usepackage[usenames,dvipsnames]{xcolor}
\usepackage[compact]{titlesec}

\usepackage{typearea}
\typearea{14}


\usepackage[linktocpage=true,
	pagebackref=true,colorlinks,
	linkcolor=BrickRed,citecolor=blue,hypertexnames=false]
	{hyperref}
	\usepackage{latexsym,graphicx,epsfig,color}
\usepackage{amsfonts,amssymb,amsmath,amsthm,amstext}
\usepackage{libertine}
\usepackage[libertine]{newtxmath}
\usepackage{enumitem}
\setlist{topsep=3pt, itemsep=0pt, leftmargin=*}
\usepackage{url,setspace}
\usepackage{multirow}
\usepackage[most]{tcolorbox}
\usepackage{rotating}
\usepackage{makeidx}
\usepackage{accents}
\usepackage{xspace}
\usepackage{algorithm}
\usepackage[noend]{algpseudocode}
\usepackage{bm}
\usepackage{changepage} 	
\usepackage{thmtools,thm-restate}
\usepackage{cleveref}
\usepackage{nicefrac}

\Crefname{algocf}{Algorithm}{Algorithms}
\crefname{algocfline}{line}{lines}
\Crefname{invariant}{Invariant}{Invariants}
\Crefname{claim}{Claim}{Claims}
\Crefname{subclaim}{Subclaim}{Subclaims}

\newcommand{\IGNORE}[1]{}
\allowdisplaybreaks

\makeindex

\newtheorem{theorem}{Theorem}[section]
\newtheorem{claim}[theorem]{Claim}

\newtheorem{lemma}[theorem]{Lemma}

\theoremstyle{definition}

\newtheorem{definition}[theorem]{Definition}

\usepackage{thmtools,thm-restate} 

\newcommand{\E}{\mathbb{E}}

\newcommand{\poly}{\operatorname{poly}}

\renewcommand{\emptyset}{\varnothing}

\newcommand{\R}{\mathbb{R}}

\newcommand{\floor}[1]{{\left\lfloor#1\right\rfloor}}
\newcommand{\ceil}[1]{{\left\lceil#1\right\rceil}}

\newcommand{\rmbig}{{\mathrm{big}}}
\newcommand{\rmsmall}{{\mathrm{small}}}
\newcommand{\fail}{{\mathrm{fail}}}

\newcommand{\bfx}{{\mathbf{x}}}
\newcommand{\bfy}{{\mathbf{y}}}

\newcommand{\Z}{\mathbb{Z}}

\newcommand{\calA}{{\mathcal A}}
\newcommand{\calC}{{\mathcal C}}
\newcommand{\calI}{{\mathcal I}}

\newcommand{\bfT}{\mathbf{T}}
\newcommand{\bfr}{\mathbf{r}}
\newcommand{\bfV}{\mathbf{V}}

\newcommand{\rmleft}{{\mathrm{left}}}
\newcommand{\rmright}{{\mathrm{right}}}

\usepackage[%
    font={small,sf},
    labelfont=bf,
    format=hang,    
    format=plain,
    margin=0pt,
    width=0.8\textwidth,
]{caption}
\usepackage[list=true]{subcaption}

\newcommand{\cost}{{\mathrm{cost}}}
\newcommand{\height}{{\mathsf{height}}}

\usepackage{mathtools,zref-savepos,xspace}

\newcommand{\bmat}{\ensuremath{\mathsf{bMat}}\xspace}
\newcommand{\obmp}{\ensuremath{\mathsf{ObMwRC}}\xspace}
\newcommand{\ogf}{\ensuremath{\mathsf{OGNF}}\xspace}
\newcommand{\olb}{\ensuremath{\mathsf{OLBwR}}\xspace}

\newcommand{\Opt}{\mathbf{Opt}}

\title{Online Unrelated-Machine Load Balancing and Generalized Flow with Recourse}

\author{
	Ravishankar Krishnaswamy\thanks{	(rakri@microsoft.com) 	Microsoft Research India.	}
		\and Shi Li \thanks{ (shil@cs.buffalo.edu) 		Department of Computer Science and Engineering, 		Universify of Buffalo.}
		\and Varun Suriyanarayana \thanks{        (varun.suriyanarayana@cornell.edu)        Cornell University.}
}

\begin{document}



\date{}

\maketitle







\thispagestyle{empty}
\begin{abstract}
We consider the online unrelated-machine load balancing problem with recourse, where the algorithm is allowed to re-assign prior jobs.  We give a $(2+\epsilon)$-competitive algorithm for the problem with $O_\epsilon(\log n)$ amortized recourse per job. This is the first $O(1)$-competitive algorithm for the problem with reasonable recourse, and the competitive ratio nearly matches the long-standing best-known offline approximation guarantee. 
We also show an $O(\log\log n/\log\log\log n)$-competitive algorithm for the problem with $O(1)$ amortized recourse. 
The best-known bounds from prior work are $O(\log\log n)$-competitive algorithms with $O(1)$ amortized recourse due to \cite{gupta2014maintaining}, for the special case of the restricted assignment model.  

Along the way, we design an algorithm for the online generalized network flow problem (also known as network flow problem with gains) with recourse. In the problem, any edge $uv$ in the network has a gain parameter $\gamma_{uv} > 0$ and $\theta$-units of flow sent across $uv$ from $u$'s side becomes $\gamma_{uv} \theta$ units of flow on the $v$'th side.  In the online problem, there is one sink, and sources come one by one. Upon arrival of a source, we need to send 1 unit flow from the source. A recourse occurs if we change the flow value of an edge. We give an online algorithm for the problem with recourse at most $O(1/\epsilon)$ times the optimum cost for the instance with capacities scaled by $\frac{1}{1+\epsilon}$.  The $(1+\epsilon)$-factor improves upon the corresponding $(2+\epsilon)$-factor of \cite{gupta2014maintaining}, which only works for the ordinary network flow problem.  As an immediate corollary, we also give an improved algorithm for the online $b$-matching problem with reassignment costs.
\end{abstract}

\newpage

\setcounter{page}{1}

\section{Introduction}
\label{sec:introduction}
Load balancing is one of the fundamental problems in online algorithms, due to its real-world motivations, as well as its clean formulation leading to the development of several techniques in online algorithms. In this paper we study the power of \emph{recourse/re-assignments} for the online load balancing problem \emph{on unrelated machines}. In the \olb (Online unrelated machine Load Balancing with Recourse) problem, we are given a set $M$ of $m$ machines, and $n$ jobs $[n]$ arrive online. Job $t \in [n]$, which arrives at time $t$, if assigned to machine $i$, would induce a load of $p_{it}$ on the  machine. The goal is to assign jobs to machines to minimize the maximum load of any machine, which is the sum of loads of jobs assigned to it. The algorithm can also re-assign prior jobs, and we separately track the \emph{recourse} to be the total number of re-assignments over the course of arrivals. 

The above is a very natural problem to study, since jobs/demands typically arrive in an online manner, and moreover, real-world systems (see, e.g. \cite{alakeel09}) often \emph{migrate jobs} between servers to achieve better balance. However, since migrating jobs is a disruptive operation, we seek to minimize the total number of movements while also ensuring nearly balanced assignments at all times. We formalize the twin objectives as follows. 

\begin{definition}[$\alpha$-competitive, $\beta$-amortized recourse Algorithms]
Let $\Opt^t$ denote the maximum load of any machine in an optimal assignment for the first $t$ jobs. Then, we say that an algorithm is $\alpha$-competitive with $\beta$-amortized recourse if the maximum load over all machines of the algorithm's assignment is most $\alpha \, \Opt^t$, and the total number of re-assignments done 
over the course of the first $t$ job arrivals is at most $\beta t$.
\end{definition}

This problem has received significant attention since the 1990s. The classical online load balancing problem has the same model as \olb, with the restriction that the online algorithm's assignments are irrevocable, i.e., no recourse is allowed. A series of works introduced several elegant ideas culminating with tight $\Theta(\log m)$-competitive algorithms \cite{AspnesAFPW97}, with matching lower bound for any randomized online algorithm. The power of recourse has also been studied for the load balancing problem, albeit from the perspective of handling job departures~\cite{westbrook2000load}, where it is evident that no online algorithm can have non-trivial competitive ratios without recourse. To the best of our knowledge, \cite{gupta2014maintaining} were the first to study the power of recourse to getting improved guarantees for job arrivals (the same setting as ours), and showed that $O(1)$-amortized recourse can yield $O(\log \log mn)$-competitive algorithms for online load balancing for the \emph{restricted assignment problem} (where each job can only be assigned to a subset $N(j)$ of machines, but it has the same processing time of $p_j$ on any of these machines). This bound represents an exponential improvement when compared to the $\Omega(\log m)$ lower bounds on the competitive ratio for the same model when no recourse is allowed.  
\emph{The central focus of this paper is in trying to understand if we can get similar improvements for the \olb on unrelated machines.} We obtain the following results:

\begin{theorem} \label{thm:unrelated-constant-approx}
For any constant $\epsilon > 0$, there is an efficient deterministic $(2+\epsilon)$-competitive algorithm for \olb on unrelated machines with $O\left(\frac{\log n \log(1/\epsilon)}{\epsilon^5}\right)$-amortized recourse. 
\end{theorem}

Note that we are able to get a competitive ratio bound nearly matching the best known offline approximation factor~\cite{ST93,lenstra1990approximation} for this classical problem, with a small amount of recourse.

\begin{theorem} \label{thm:unrelated-constant-recourse}
There is an efficient randomized $O(\log \log n/ \log \log \log n)$-competitive algorithm for \olb on unrelated machines with $O(1)$-amortized recourse. 
\end{theorem}

Our algorithms work by maintaining a $(1+\epsilon)$-competitive fractional assignment online, and then rounding it in an online manner while ensuring that both steps do not modify the solution too much. Prior to our work, no algorithm that maintained constant-competitive fractional solutions with bounded recourse was known.
Indeed, a crucial observation central to our result is that the fractional assignments for unrelated machines can be seen as a special case of the \emph{generalized flow} (a.k.a min-cost flow with gains) problem. Here, we are given a directed graph $G = (V,E)$ with cost $c_e$ and capacities $\mu_e$ on the edge set $E$. However, unlike classical flows where the same amount of flow exits the edge as entering, in the generalized problem, the amount of flow exiting an edge is a scalar multiple of the amount entering. This is captured by \emph{gain factors} $\gamma_e$ on the edges, which represent the \emph{extent to which one unit of flow originating at one end-vertex of the edge gets scaled when it reaches the other end-vertex}. We are given a set $S$ of sources, each of which wants to send one unit of flow, and a sink $\tau$ which can absorb the flows. The goal is to maintain minimum cost flows which can send unit from the sources while ensuring flow conservation at all vertices $v \in V \setminus S \setminus \{\tau\}$. 
Since we seek to maintain online fractional solutions for \olb, we introduce and study the \ogf (Online Generalized Network Flow) problem. Here, the sources arrive one by one, and the algorithm needs to pay the incremental cost of sending one unit flow from the new source, on top of the existing flow. 
The goal is to minimize the total cost in comparison with the offline optimum. Informally, we obtain the following result, and state the formal version in Theorem~\ref{thm:flow}.

\begin{theorem} \label{thm:genflow}
There is an efficient deterministic $O\left(\frac1\epsilon\right)$-competitive algorithm for \ogf when the algorithm can violate edge capacities by a factor of $(1+\epsilon)$. 
\end{theorem}

In fact, the existing \olb results for restricted assignment follow a similar template. However, the reduction in~\cite{gupta2014maintaining} only applies to the restricted assignment setting, and results in an online version of the classical flow problem (without gains).~\Cref{thm:flow} represents a two-fold advancement over this. Firstly,~\cite{gupta2014maintaining} presents $O(1)$-competitive algorithms with $(2+\epsilon)$-factor capacity violations, while our results improve this to $(1+\epsilon)$ violation, and secondly, our algorithms can handle arbitrary gain factors, while the GKS-algorithms only applies to unit gains.

We now discuss our final algorithmic contribution. Indeed, since our results improve the capacity violation even for the regular flow problem, this immediately translates to an improvement for the so-called Online $b$-Matching with Reassignment Costs (\obmp) problem, where we have a bipartite graph $G = (L \uplus R, E)$. Right vertices $v \in R$ each have capacities $b_v \in \Z_{\geq 0}$, and the goal is to assign the left vertices while respecting the capacities. It is guaranteed that $G$ has a valid $b$-matching: a matching where every $u\in L$ is matched once and every $v \in R$ is matched at most $b_v$ times. In the online problem, the left vertices arrive one by one. When a left-vertex $u \in L$ arrives, it specifies its neighbors $N(u) \subseteq R$ in $G$, and a \emph{reassignment cost} $c_u \geq 0$. Each time we re(assign) a $u \in L$ to a (new) vertex $v \in R$, we pay a cost of $c_u$. We need to always maintain an assignment which violates the capacities by a small amount while incurring a small cost.

\begin{restatable}{theorem}{thmbmat}
\label{thm:bmat}
There is an efficient deterministic algorithm for \obmp which, with $O\left(\frac1\epsilon\right)$ amortized recourse, maintains a matching where each $v \in L$ is matched once and each $v \in R$ is matched at most $\ceil{(1+\epsilon)b_v}$ times. 
\end{restatable}

This improves over the prior work~\cite{gupta2014maintaining} where the capacity violation was a factor of $(2+\epsilon)$. 

Finally, one may ask what we can do in the fully-dynamic model for the online load balancing problem with recourse. That is, jobs may arrive and depart. In Appendix~\ref{appendix:FD-LB}, we show that in this case, even the offline algorithm which knows the whole sequence of the job arrivals and departures ahead of time need to incur an amortized recourse of $n^{\Omega(1/\alpha)}$, in order to achieve an $\alpha$-competitive ratio.  Thus, to circumvent the negative result, one needs to consider a different measurement for recourse for the fully dynamic model. 

\subsection{Our Techniques, at a High Level}
\noindent {\bf Load Balancing.} \cite{gupta2014maintaining} solved the restricted assignment problem using the following method: for every job $j$ with load $p_j$, they create a source $j$ of demand $p_j$; there is also a vertex $v_i$ for every machine $i$. The source $j$ has a directed edge to $v_i$ if and only if $j$ can be scheduled on $i$. Finally, there are directed edges of capacity $T^*$, the optimal makespan bound, from each $v_i$ to a common sink $\tau$. It is then easy to see that feasible flows in this network correspond to fractional solutions for the restricted assignment problem. \cite{gupta2014maintaining} therefore use their online flow algorithm to them maintain $O(1)$-competitive fractional solutions with $O(1)$-recourse. They then \emph{round} the fractional solution using a combination of randomized rounding and bucketing jobs based on their size to obtain the $O(\log \log mn)$-competitive ratio. However, note that both steps (computing the fractional solution as well as bucketing) crucially use the fact that a job has equal processing time $p_j$ on all its valid machines, and it is not clear how to extend either to the unrelated machines setting.

To circumvent these issues, we instead view the fractional solution to the unrelated machines problem as a \emph{generalized flow} instance, which is what motivated our study of the \ogf problem in the first place.
We can then immediately use~\Cref{thm:genflow} to maintain $(1+\epsilon)$-competitive fractional solutions with $O\left(\frac1\epsilon\right)$-recourse.  As for the rounding step, we first present a cleaner algorithm with $O(1)$-approximation ratio, that delivers the key ideas. We view each machine as $O(\log n)$ sub-machines, with sub-machine $v_{ik}$ only catering to jobs $j$ of size $p_{i,j} \approx T^*/2^k$. Then, suppose that in some fractional solution $\{x_{i,j}\}$, a total of $f_{ik}$ units of jobs is assigned to sub-machine $v_{ik}$. Then, if we are able to assign roughly $\ceil{f_{ik}}$ jobs in our rounded solution, we can guarantee that our load will not exceed that of the fractional solution by too much!
With this intuition, we can solve the rounding challenge as follows: we maintain an instance of online $b$-matching where jobs correspond to left vertices, and sub-machines $v_{ik}$ correspond to right vertices. Each right vertex also has a capacity \emph{equal to the total fractional allocation in the online LP solution, rounded up to the next integer}. Then, when the fractional solution is updated, we update the capacities of right vertices, and simply run a low-recourse algorithm for the online $b$-matching instance to recover our overall assignment. 

The online rounding algorithm with $(2+O(\epsilon))$-approximation ratio is more involved. We use the grouping idea of \cite{ST93} for the $(2+\epsilon)$-approximation for the generalized assignment problem. We maintain a 2-level partition of fractional jobs assigned to every machine $i$, according to non-increasing order of job sizes: They are first partitioned into \emph{buckets}, each of which contain roughly $\Theta(1/\epsilon)$ factional jobs, and then every bucket is partitioned into \emph{segments}, each of which contains $1-\Theta(\epsilon)$ fractional jobs. The bipartite matching instance can be constructed from the segments: Each segment corresponds to a vertex on the right side. A good fractional matching is guaranteed to exist by existence of the fractional assignments. As the segments are created by job sizes, a matching gives a $(2+O(\epsilon))$-approximate assignment, as in \cite{ST93}.

Our $L(n) = O\left(\frac{\log \log n}{\log \log \log n}\right)$-competitive algorithm with $O(1)$-amortized recourse in Theorem~\ref{thm:unrelated-constant-recourse} uses many ideas from the $L(m)$-competitive online rounding algorithm of \cite{LX21}. The offline version of our algorithm works as follows. A job $j$ is big on machine $i$ if $p_{ij} > T^*/\log n$ and small otherwise. Then a job is big if $1/2$ fractional of it is big on its assigned machines in the fractional solution, and small otherwise. For small jobs, independent rounding already gives an $O(1)$-approximation ratio. For big jobs $j$, with first round $x_{ij}$ values so that each $x_{ij}$ becomes either $0$ or at least $1/\log n$. This makes the support of $x$ for big jobs sparse.  Then we attempt to assign a big job to a machine randomly using the new $x_{ij}$ values. The assignment fails if the target machine is too overloaded. Finally, we use the deterministic 2-approximation rounding algorithm to round the failed jobs. To analyze the recourse, we use two crucial lemmas: A job $j$ fails with $1/\poly\log (n)$ probability, and with high probability, a connected component induced by failed jobs has size $\poly\log(n)$. Therefore, even if we reassign everything in such a connected component, the recourse is small. 

\noindent {\bf Generalized Flow.} Indeed \cite{gupta2014maintaining} studied the special case of this problem with unit gains, and showed that the natural algorithm of satisfying the requirement of a new source by sending the unit flow along the shortest path to the sink in the residual graph, actually works. Here, the residual graph might have some forward arcs and some backward arcs, but the GKS algorithm, assigns a cost of $c_e$ for forward arcs as well as backward arcs in the residual graph, since that is the cost of the augmentation incurs in their online model. This is in contrast to the classical residual graph for offline flow problems where backward arcs have opposite sign of costs. Now, they define a quantity ${\sf height}(s)$, which is the \emph{cost of the shortest path to the sink in the residual graph}, and show that it is a good estimate of the augmentation cost of $s$ when it arrives. Moreover, they show that ${\sf height}(s)$ is non-decreasing over arrival of subsequent sources. Finally, they show how we can construct a good dual solution to the offline LP formulation of the min-cost flow problem using the ${\sf height}(s)$ values, using which they relate the online and offline costs.

How do we extend these ideas to the generalized flow problem? Intuitively, when a new source arrives, a natural strategy might be to augment along the shortest path in the residual graph. While this is a reasonable idea, observe that the concept of augmenting paths is very different when there are gains. Indeed, we could potentially augment along a path in the residual graph from the source to a cycle which does not even contain the sink, provided the product of gains on the cycle is $< 1$ (these are called \emph{flow-absorbing cycles} in literature). 
Keeping this in mind, we generalize the notion of height above to be the \emph{minimum cost way to send one unit flow out of the source}, as opposed to the cost of the shortest path from source to sink. We then show that with these modifications, we can seamlessly apply the duality-based proof technique of \cite{gupta2014maintaining}. Finally, to improve the capacity violation factor to $(1+\epsilon)$, we actually set the cost of backward arcs to $0$ in the residual graph as opposed to a positive number. Up to a factor of 2 in the recourse, the two definitions are equivalent. This is also the reason that in the definition of our problem, we only incur costs for increasing flow values. However, this minor modification in defining the heights helps us perform a tighter analysis to get the improved the capacity violation bound.

\subsection{Related Work}
It is well known that without any form of recourse, the best possible competitive ratio for online load balancing is $O(\log n)$ \cite{AzarNR95}. When arrivals and departures are allowed, \cite{AzarBK94} give a lower bound of $O(\sqrt{n})$. Philip and Westbrook~\cite{PhillipsW93} considered the same case and showed $O(\log n)$-competitive algorithms with $O(1)$-recourse for the special case of unit-size restricted assignment (aka online matching). Westbrook~\cite{westbrook2000load} subsequently designed $O(1)$-competitive algorithms with $O(\log n)$-recourse for the same setting. The case of unrelated machines in the fully-dynamic case was considered in~\cite{AwerbuchAPW01} where the authors designed $O(\log n)$-competitive algorithms with $O(\log n)$-recourse. In the special case of identical machines, \cite{RudinC03} gave a lower bound $\sqrt{3}$ and an upper bound of $1.923$ is known from \cite{Albers99}. \cite{SandersSS09} presented a class of $(1+\epsilon)$-competitive algorithms with reassignments allowed, where the mitigation factor grows as $\epsilon$ gets smaller. When there are only arrivals and no departures, but recourse is allowed, \cite{gupta2014maintaining} showed that one can achieve a $O(\log \log nm)$ competitive ratio with $O(1)$-amortized recourse in the restricted assignment setting. \cite{BerndtEM22} extended this result to general re-assignment costs. Both papers also gave constant approximations with similar amount of reassignments in the special case of OGNF with all gains $\gamma_e = 1$. 

There is also significant work on designing offline algorithms for generalized flow problems (see e.g., Chapter 6 of ~\cite{flowbook}), and these algorithms have also been used to design fast approximation algorithms for makespan minimization~\cite{GAIRING}. We believe our work is the first to address the online version of generalized flow.

\paragraph{Organization:} In Section~\ref{sec:flow} we present the $O(\frac{1}{\epsilon})$ competitive online algorithm for generalized flow with a $(1+\epsilon)$-capacity violation, proving Theorem~\ref{thm:genflow}. Using this algorithm, in Section~\ref{sec:lb} we design an $O(1)$-competitive algorithm with $O(\log n)$ amortized recourse; the $(2+\epsilon)$-competitive rounding algorithm that proves Theorem~\ref{thm:unrelated-constant-approx} is deferred to Appendix~\ref{sec:2-approx-rounding}.  In Section~\ref{sec:const-recourse}, we present the $O(\frac{\log \log n}{\log \log \log n})$-competitive algorithm for the same problem with $O(1)$ amortized recourse, proving Theorem~\ref{thm:unrelated-constant-recourse}. In Appendix~\ref{appendix:FD-LB} we give the lower bound on the recourse required in the fully dynamic model. Missing proofs can be found in Appendix~\ref{appendix:missing-proofs}.

\paragraph{Notations} For a real number $a$, we use $(a)_+$ to denote $\max\{a, 0\}$. For a graph $H$ and a vertex $v$ in $H$, we shall use $N_H(v)$ to denote the set of neighbors of $v$ in $H$. For a directed graph $H$ and a vertex $v$ in $H$, we use $\delta^+_H(v)$ and $\delta^-_H(v)$ to denote the sets of outgoing and incoming edges of $v$ in $H$ respectively. When $H$ is $G$ in the context, we omit the subscript.  For online load balancing, we shall identify jobs with time steps: jobs are denoted as $[n]$, where job $t \in [n]$ arrives at time $t$.



\section{The Online Generalized Network Flow (\ogf) Problem}
\label{sec:flow}
In the generalized network flow problem, we are given a digraph $G=(V,E)$ with sources $S \subseteq V$ and a sink $\tau \in V \setminus S$, where the sources $S$ do not have incoming edges and the sink $\tau$ does not have outgoing edges in $G$.  
We are given vectors $\mu, c, \gamma \in \R_{>0}^E$, where for every $e \in E$, $\mu_e$, $\gamma_e$ and $c_e$ denote the capacity, gain and cost of the edge $e$ respectively. Every source $s \in S$ has $a_s > 0$ units of supply.  As in the ordinary network flow problem, our goal is for each $s \in S$ to send $a_s$ units flow in the network satisfying the flow conservation and edge capacity constraints, so as to minimize the cost. The generalization comes from the gain vector $\gamma$: when a vertex $u$ sends $\theta$ unit flow along an edge $e = uv$, $v$ will receive $\gamma_e \theta$ units flow from the edge.  Therefore, the problem can be formulated as the following LP, where we assume $a_v = 0$ for every $v \in V \setminus S \setminus \{\tau\}$.
\begin{equation}
	\min \sum_{e \in E} c_e x_e  \qquad \text{s.t.} \label{LP:generalized}
\end{equation}\vspace*{-12pt}
\begin{align*}	
	\sum_{e \in \delta^{+}(v)} x_e -\sum_{e \in \delta^{-}(v)} \gamma_e x_e&=a_v &\quad &\forall v \in V \setminus \{\tau\}\\
	0 \leq x_e &\leq \mu_e &\quad &\forall e \in E
\end{align*}
In the LP, $x_e$ for an edge $e = uv \in E$  indicates the amount of flow sent by $u$ through $e$. Both the capacity $\mu_e$ and the per-unit cost $c_e$ are defined w.r.t the flow on the sender's side of $e$.   We call any $\bfx \in \R_{\geq 0}^E$ satisfying the constraints in the LP a valid flow for the instance.  Notice that unlike the ordinary network flow problem, the balance parameter (i.e., the $a$-value) of the sink $\tau$ can not be decided by those of other vertices. Therefore we leave $a_\tau$ undefined and do not impose the flow conservation constraint on $\tau$.

\noindent {\bf Online Generalized Network Flow Problem (\ogf)}: We are initially given an instance of the problem $(G = (V, E), S, \tau, \mu, c, \gamma)$ with $S = \emptyset$.  In each time step $t = 1, 2, \cdots, T$, a new source $s_t \notin V$ arrives. We assume $a_{s_t} = 1$, as this suffices for our purpose. Along with $s_t$, we are given the outgoing edges of $s_t$, as well as their $\mu, c$ and $\gamma$ values.  After the arrival of $s_t$, we add it to $S$ and $V$, and its outgoing edges to $E$. 


For any time $t$, we need to maintain a valid flow $\bfx^{(t)} \in \R_{\geq 0}^E$ for the instance at time $t$. The cost incurred at this time step is defined as $\sum_{e \in E} c_e\big(x^{(t)}_e - x^{(t-1)}_e\big)_+$, where we assume undefined variables have value $0$. Our goal is to design an online algorithm with a small cost. 

We elaborate more on the definition of our cost; for better clarity, we focus on the ordinary network flow problem (that is, we have $\gamma_e = 1$ for all $e \in E$).  
 \cite{gupta2014maintaining} defined the cost incurred at step $t$ to be $\sum_{e \in E} c_e\cdot |x^{(t)}_e - x^{(t-1)}_e|$ in their model to accurately capture the notion of recourse for flow problems. In this cost model, decreasing the flow value along an edge would incurs a positive cost. This is in contrast to classical offline algorithms for flow where decreasing the flow value of an edge would incur a negative cost, thereby ensuring that the final cost of the flow would always be equal to the sum of costs incurred by each update. Our model is in between these models, where we charge flow increases with positive cost but omit charging flow decrements. However, note that our results also translate to results for the~\cite{gupta2014maintaining} cost model within a factor of $2$ in the cost  --- indeed, for any decrement in the flow value on an edge, we must have paid the cost when we increase the flow value. 

We prove the following main theorem in this section, which is a more formal description of Theorem~\ref{thm:genflow}: 
\begin{theorem}
	\label{thm:flow}
	Given any $\epsilon > 0$, there is an efficient deterministic algorithm for \ogf, such that the following holds. The cost incurred by the algorithm at any time is at most $\frac{1+\epsilon}{\epsilon} = O\big(\frac1\epsilon\big)$ times the cost of the optimum flow for the general network flow instance at the time, with all edge capacities scaled by $\frac{1}{1+\epsilon}$.
\end{theorem}

For convenience, we assume there are dummy edges from every vertex in $V \setminus \{\tau\}$ to $\tau$ with infinite capacity, cost $B$ for a sufficiently large $B$, and gain $1$, to always ensure feasibility in our definitions. When the original instance is feasible and $B$ is large enough, our algorithm will not use the dummy edges. 

\subsection{Augmenting Paths for General Network Flow}
Our algorithm for online generalized network flow can be defined in a straightforward manner: When a source $s_t$ arrives, we repeatedly find the cheapest ``generalized augmenting path'' from $s_t$ and augment the flow along the path as much as possible. The procedure terminates when $s_t$ sent 1 units of flow. Unlike the ordinary network flow problem, for which we can define an augmenting path as a simple path from $s_t$ to $\tau$, an augmenting path for the general network flow problem can be slightly more complex.

\begin{definition}[Residual Graph]
	\label{def:residual}
	Let $(G, S, \tau, \mu, c, \gamma)$ be a generalized network flow instance.  Let $\bfx \in \R_{\geq 0}^E$ be a vector such that $x_e \in [0, \mu_e]$ for every $e \in E$ (it is not required that $\bfx$ is a valid flow for the instance).  Then the \emph{residual graph} $G^\bfx = (V, E^\bfx)$ for $\bfx$ is defined as the graph containing vertices $V$ and the set $E^\bfx$ of edges, each $e \in E^\bfx$ with parameters $\mu^\bfx_e, c^\bfx_e$ and $\gamma^\bfx_e$. They are defined as follows. 
	\begin{itemize}
		\item For every $uv \in E$ with $x_{uv} < \mu_{uv}$, we have a \emph{forward} edge $uv \in E^\bfx$ with $\mu^\bfx_{uv} = \mu_{uv} - x_{uv}$, $c^\bfx_{uv} = c_{uv}$ and $\gamma^\bfx_{uv} = \gamma_{uv}$. 
		\item For every $uv \in E$ with $x_{uv} > 0$, we have a \emph{backward} edge $vu \in E^\bfx$ with $\mu^\bfx_{vu} = \gamma_{uv}x_{uv}$, $c^\bfx_{vu} = 0$ and $\gamma^\bfx_{vu} = 1/\gamma_{uv}$.
	\end{itemize}
\end{definition}
The definition can be viewed as an extension of the residual graph to the generalized network flow problem. We make two remarks here. First, a backward edge has cost $0$ instead of a negative cost, due to the cost defined in our online model. Second, $\theta$ units of flow sent via $uv$ on the $u$'s side transform into $\gamma_{uv}\theta$ units of flow on the $v$'th side. This gives the definition of $\mu^\bfx_{vu}$ and $\gamma^\bfx_{vu}$ for a backward edge $vu \in E^\bfx$. As is typical, we assume $G$ does not contain anti-parallel edges, so that $G^\bfx$ is a simple graph.

\begin{definition}[Fractional Augmenting Paths]
	\label{def:fractiona-augmenting}
	Let $(G, S, \tau, \mu, c, \gamma)$, $\bfx$, $G^\bfx = (V, E^\bfx)$ and vectors $\mu^\bfx, \gamma^\bfx$ and $c^\bfx$ be defined as in Definition~\ref{def:residual}.  Let $s \in V \setminus \{\tau\}$ (it may be possible that $s \notin S$).  A fractional augmenting path from $s$ in $G^\bfx$ is a vector $f \in \R_{\geq 0}^{E^\bfx}$ such that the excess flow $\sum_{e \in \delta^+_{G^\bfx}(v)} f_e - \sum_{e \in \delta^-_{G^\bfx}(v)} \gamma^\bfx_e f_e$ equals $1$ for $v = s$, and equals $0$ for $v \in V \setminus \{s, \tau\}$. 
	The cost of such an augmenting path $f$ is defined as $\cost(f) := \sum_{e \in E^\bfx}c^\bfx_e f_e$.
\end{definition}
Notice that the definition does not involve the capacities $\mu^\bfx_e$ of edges $e \in E^\bfx$: As long as $e$ exists in $E^\bfx$, $f_e$ can take any number in $\R_{\geq 0}$.

\begin{definition}[Augmenting Paths]
	\label{def:augmenting}
	Take all the notations in Definition~\ref{def:fractiona-augmenting} and assume $f$ is a fractional augmenting path from $s$ in $G^\bfx$. We simply say $f$ is an augmenting path  (without the word ``fractional'') if additionally it satisfies one of the following conditions:
	\begin{enumerate}[label=(\ref{def:augmenting}\alph*)]
		\item The support of $f$ is a path from $s$ to $\tau$ in $G^\bfx$.  \label{item:AP-path}
		\item The support of $f$ is a cycle $C$ in $G^\bfx$ containing $s$ but not $\tau$. \label{item:AP-cycle}
		\item The support of $f$ is the union of a cycle  $C$ in $G^\bfx$ not containing $s$ and $\tau$,  and a path in $G^\bfx$ from $s$ to $C$ that is internally disjoint from $C$. \label{item:AP-union}
	\end{enumerate}
\end{definition}

\begin{restatable}{lemma}{lemmanonfractional}
	\label{lemma:non-fractional}
	Consider a generalized network flow instance $(G, S, \tau, \mu, c, \gamma)$ and $\bfx \in \R_{\geq 0}^E$ satisfying that $x_e \in [0, \mu_e]$ for every $e \in E$. Let $v \in V \setminus \{\tau\}$. Then $f$, the minimum-cost fractional augmenting path  from $v$ in $G^\bfx$, assuming it exists, can be achieved at an augmenting path.  
\end{restatable}

The following claim will be useful later.
\begin{claim}
	\label{claim:break}
	Let $f$ be an augmenting path from some $s \in V \setminus \{\tau\}$ in $G^\bfx$ for some $\bfx$. Let $v \notin \{s, \tau\}$ be some vertex in the support graph of $f$. Then we can break $f$ into $f = f' + f''$ where $f'$ is a flow path in $G^\bfx$ that sends 1 unit flow from $s$ to $v$ (the flow received by $v$ may not be $1$), and $f''$ is a scaled augmenting path from $v$ in $G^\bfx$. 
\end{claim}

Now given a (fractional) augmenting path $f$ from $s$ in $G^\bfx$, and a real number $\theta > 0$,  we define the operation of augmenting $\bfx$ by $\theta$ units using $f$ as follows:
\begin{itemize}
	\item For every forward edge $uv \in E^\bfx$ with $f_{uv} > 0$, we update $x_{uv} \gets x_{uv} + \theta\cdot f_{uv}$. 
	\item For every backward edge $vu \in E^\bfx$ with $f_{vu} > 0$, we update $x_{uv} \gets x_{uv} - \theta \cdot f_{vu} / \gamma_{uv}$. 
\end{itemize}

\begin{claim}
	Let $f$ be an augmenting path from $s$ in $G^\bfx$. Then, augmenting $\bfx$ by $\theta > 0$ units using $f$ does not change $\sum_{e \in \delta^+(v)} x_e - \sum_{e \in \delta^-(v)} \gamma_e x_e$ for every $v \in V \setminus \{s, \tau\}$, and increases $\sum_{e \in \delta^+(s)} x_e - \sum_{e \in \delta^-(s)} \gamma_e x_e$ by $\theta$. 
\end{claim}


\subsection{The Online Algorithm}
With the definition of augmenting paths, we can now formally describe our online algorithm.  The pseudo-code is given in Algorithm~\ref{alg:general-online}.
\begin{algorithm}
	\caption{Online algorithm for generalized network flow}
	\label{alg:general-online}
	\begin{algorithmic}[1]
		\State Let $\bfx^{(0)}$ be the all-0 vector over edges of the initial graph $G$
		\For{every $t \gets 1$ to $T$}
		\State update the instance to include the source $s_t$
		\State let $\bfx \gets \bfx^{(t-1)}$, adding $0$-coordinates for the incoming edges of $s_t$
		\While{$\sum_{e \in \delta^+(s_t)} x_e < 1$}
		\State find the cheapest augmenting path $f$ from $s_t$ in the residual graph $G^\bfx$ \label{step:find-path}
		\State let $\theta > 0$ be the biggest number such that after augmenting $\bfx$ using $f$ by $\theta$ units, we still have $x_e \in [0, \mu_e]$ for every $e \in E$ and $\sum_{e \in \delta^+(s_t)} x_e \leq 1$
		\State augment $\bfx$ using $f$ by $\theta$ units
		\EndWhile
		\State $\bfx^{(t)} \gets \bfx$ 
		\EndFor
	\end{algorithmic}
\end{algorithm}

{\bf Remark on Running Time} \ Notice that the running time of Algorithm~\ref{alg:general-online} may be exponentially large, as we can not bound the number of iterations of the while loop by polynomial. However, at time $t$, the algorithm is simply trying to send 1 unit flow from $s_t$ using the minimum cost in the residual graph, and this can be simply done using an LP. Moreover sending the 1 unit flow at once can only incur a smaller cost.  Therefore we can make the algorithm efficient.  For analysis purposes, we chose to present Algorithm~\ref{alg:general-online}, since the augmenting path at each step has a good structure as stated in Definition~\ref{def:augmenting}.

\subsection{Analysis of Online Algorithm}

\cite{gupta2014maintaining} define the concept of height of any vertex to be the cost of the shortest path to the sink $\tau$ in the residual graph at any point in time, and then track it for bounding the competitive ratio. Analogously, we define a height for a vertex as the \emph{cost of the cheapest augmenting path from the vertex in the residual graph} -- notice the lack of sink in the definition. We show that the incremental cost incurred by the $t^{\rm th}$ source is at most the height of $s_t$ at time $t$. Moreover, we also show that heights can only increase during the course of the algorithm, implying that the cost incurred by the algorithm is at most the total height of sources at the end.  To complete the analysis, we show that the heights (after all arrivals) define a dual solution for the final flow instance with capacities scaled down by  a factor of $(1+\epsilon)$, and upper bound its cost by $O(1/\epsilon)$ times the offline optimum cost of the instance.

\begin{definition}
	Let $t \in [0, T]$ and $v \neq \tau$ be a vertex in the network $G$ at time $t$. The \emph{height} of $v$ at time $t$, denoted as $\height_t(v)$, is defined as the cost of the cheapest augmenting path from $v$ in the residual graph $G^{\bfx^{(t)}}$.  We define $\height_t(\tau) = 0$.
\end{definition}

\noindent {\bf Monotonicity of Heights.} We  show that heights are non-decreasing over the course of the algorithm. The proof the following lemma is deferred to Appendix~\ref{appendix:missing-proofs}. 
\begin{restatable}{lemma}{lemmaheightincreases}
	\label{lemma:height-increases}
	Let $t \in [T]$ and $v$ be a vertex in the graph $G$ at time $t-1$. Then $\height_t(v) \geq \height_{t-1}(v)$.
\end{restatable}

\begin{lemma}
	\label{lemma:bound-cost-by-heights}
	The cost incurred by the whole algorithm is at most $\sum_{t = 1}^T \height_T(s_t)$. 
\end{lemma}
\begin{proof}
	Notice that the cost incurred by sending 1 unit of flow from the source $s_t$ in time $t$ can be upper bounded by $\height_t(s_t)$, which in turn is at most $\height_T(s_t)$ by Lemma~\ref{lemma:height-increases}.
\end{proof}

\noindent {\bf Bounding total heights using duality.}
Let $C^*$ be the cost of the optimum flow for the generalized network flow instance at time $T$, with capacities scaled by a factor of $\frac{1}{1+\epsilon}$.  We consider the dual of LP \eqref{LP:generalized} for the instance: 
\begin{equation}
	\max \qquad  \sum_{v \in V \setminus \{\tau\}}a_v y_v- \sum_{e \in E}\frac{\mu_ez_e}{1+\epsilon}  \qquad \text{s.t.} 
\end{equation}\vspace*{-15pt}
\begin{align*}
	 - z_{uv} + y_u - \gamma_{uv}y_v &\leq c_{uv} &\quad &\forall uv \in E \\
	z_{e} &\geq 0 &\quad &\forall e \in E\\
	y_{\tau} &=0
\end{align*}
Notice that we do not have a constraint in LP \eqref{LP:generalized} for $\tau$; for convenience we also introduce a dual variable $y_\tau$ and let $y_\tau = 0$.  For a fixed $y$ vector,  the optimal choice for $z_{uv}$ is $(y_u  - \gamma_{uv}y_v  - c_{uv})_+$. Also $a_v = 0$ for every $v \in V \setminus S \setminus \{\tau\}$ and $a_{s_t} = 1$ for every $t \in [T]$. Therefore, the dual LP can be rewritten as 
\begin{align*}
	\max \qquad \sum_{t \in [T]} y_{s_t}-\sum_{uv \in E}\frac{\mu_{uv} \left(y_u - \gamma_{uv}y_v-c_{uv}\right)_+}{1+\epsilon}  
	\qquad \text {s.t. } \quad y_\tau=0.
\end{align*}

\begin{lemma}
	\label{lemma:bound-heights-using-dual}
	$\sum_{t=1}^T \height_T(s_t) \leq \frac{1+\epsilon}{\epsilon}C^*$.
\end{lemma}
\begin{proof}
	Now we need to bound the sum of the heights of the sources at termination. To do this we show that $\big(y_v := \height_T(v)\big)_{v \in V}$ is a feasible dual solution. Then we bound the cost of this feasible dual in relation to $C^*$, therefore giving us a bound on the competitive ratio.
	
	
	Let $\bfx  = \bfx^{(T)} \in \R_{\geq 0}^E$ be the final flow we obtained. By breaking edges, we can assume every edge $e \in E$ has either $x_e = 0$ or $x_e = \mu_e$. Any edge $uv \in E$ with $x_{uv} = 0$ exists in the residual graph $G^\bfx$. Therefore $y_u \leq \gamma_{uv} y_v+c_{uv}$ as $y$ corresponds to the heights, and sending 1 unit flow from $u$ can be achieved by sending $1$ unit flow from $u$ to $v$, and then sending $\gamma_{uv}$ units flow from $v$. So for such edges $(y_u - \gamma_{uv} y_v - c_{uv})_+ = 0$.  Let $E'$ be the set of edges with $x_{uv} = \mu_{uv}$. So, 
	\begin{align*}
		\sum_{uv \in E} \frac{\mu_{uv}(y_u - \gamma_{uv} y_v - c_{uv})_+}{1+\epsilon} \leq \sum_{uv \in E'} \frac{x_{uv}}{1+\epsilon} (y_u - \gamma_{uv} y_v) = \frac{1}{1+\epsilon}\sum_{v \in V} y_v\left(\sum_{e \in \delta^+(v)} x_e - \sum_{e \in \delta^-(v)} \gamma_ex_e\right)=\frac{1}{1+\epsilon}\sum_{t \in [T]}y_{s_t}.
	\end{align*}
	For every $uv \in E'$, $x_{uv} = \mu_{uv}$  and $y_v \leq \frac{y_u}{\gamma_{uv}}$ since the backward edge $vu$ exists in $G^\bfx$ and it has gain $\frac1{\gamma_{uv}}$ and cost $0$. So we have the inequality.  The first equality is by that $x_{uv} = 0$ for $uv \in E \setminus E'$ and rearranging the terms. The second equality is by the balance condition for $\bfx$ and $y_\tau = 0$.

	So the objective value of the dual solution $y$ is at least $(1-\frac{1}{1+\epsilon})\sum_{t \in [T]} y_{s_t}$, which implies $\frac{\epsilon}{1+\epsilon} \sum_{t \in [T]} y_{s_t} \leq C^*$.  Multiplying both sides by $\frac{1+\epsilon}{\epsilon}$ proves the lemma.
\end{proof}

Thus, combining Lemmas~\ref{lemma:bound-cost-by-heights} and \ref{lemma:bound-heights-using-dual} proves Theorem~\ref{thm:flow}.


\section{Online Unrelated Machine Load Balancing with Recourse} 
\label{sec:lb}
We now show one of our main results, that of maintaining $(2+\epsilon)$-approximate solutions for online unrelated machine load balancing with $O_\epsilon(\log n)$ amortized recourse per job, as stated in Theorem~\ref{thm:unrelated-constant-approx}. We restate the problem setting. There is a set $M$ of $m$ machines, and $n$ jobs indexed by $[n]$. We have a bipartite graph $G = (M \uplus [n], E)$ between machines and jobs, where $ij \in E$ indicates that the job $j$ can be assigned to machine $i$. When $j$ is assigned to $i$, it incurs a load of $p_{ij} > 0$ on machine $i$.  The goal is to assign jobs to machines so as to minimize the maximum load over all machines, also called makespan in the scheduling literature. In the online version, jobs arrive one by one: job $j \in [n]$ arrives at time $j$, along with its incident edges in $G$ and their $p_{ij}$ values.  We need to maintain a solution for the arrived jobs at any time.  We allow the algorithm to re-assign prior jobs from time to time, and separately track the recourse of the algorithm. 


\noindent {\bf Known vs Unknown $T^*$.} 
We first define a useful quantity $T^*$, which is the smallest value of $T$ for which the following LP is feasible: $\sum_{i \in N(j)}x_{ij} = 1$ for every $j \in [t]$, $\sum_{j \in N(i)} p_{ij}x_{ij} \leq T$ for every $i \in M$, and $x_{ij} = 0$ if $p_{ij} > T^*$. We refer to this as the \emph{optimal fractional makespan}. 

Our online algorithms will assume knowledge of $T^*$. While we can use a standard guess-and-double approach to eliminate this assumption, we would lose an additional constant factor in the competitive ratio. Since we are allowed recourse, we can do better, as follows to get the $2+O(\epsilon)$ guarantee on competitive ratio. Suppose there is algorithm $\calA$ that achieves $C \cdot T^*$ makespan when $T^*$ is given.  We now design a simple procedure which can also achieve $(1+O(\epsilon)) C$-competitive solutions with bounded recourse, even when $T^*$ is not given. Indeed, we break our procedure into \emph{stages}, where a new stage occurs when the optimum fractional makespan increases by a factor of at least $1/\epsilon$. Each stage is further partitioned into many \emph{phases}, where a new phase starts if the optimum fractional makespan increases by a factor of at least $1+\epsilon$. When a new stage $g$ starts with bound $T^*$ on the optimal fractional makespan, we simply re-construct an offline solution for all the jobs in $(g-2)$-th stage with makespan $2\epsilon T^*$, say using the $2$-approximation algorithm~\cite{lenstra1990approximation} for offline load balancing. We then \emph{freeze the assignment of these jobs} according to this offline solution, i.e., we won't change the assignment of these jobs in the future. Note that the total load due to all the frozen jobs on any machine is at most $O(\epsilon) T^*$. On the other hand, whenever a new phase starts with optimum fractional makespan $T^*$, we re-run algorithm $\calA$ with the revised estimate $T^*$ and reintroduce all the unfrozen jobs (i.e., jobs of this stage as well as the previous), thereby causing recourse for all these jobs. From the gaurantee of $\calA$, the makespan induced by these jobs will be at most $C \cdot T^*$, giving us the desired guarantee. As for recourse, note that each job can be unfrozen for at most $O(\log_{1+\epsilon}\frac1\epsilon) = O\left(\frac{\log (1/\epsilon)}{\epsilon}\right)$ phases across two stages, which bounds the recourse.



\noindent {\bf Algorithm Overview.} 
 The overall algorithm comprises of two steps. We first maintain a \emph{fractional} solution which is $(1+\epsilon)$-competitive and has $O\big(\frac1\epsilon\big)$-amortized fractional recourse by reducing the problem to the online generalized network flow problem. In the second step we round the fractional solution into an integral one, in an online manner with low recourse by creating an intermediate bipartite matching instance based on the fractional solution. For clarity of presentation, we present a conceptually simpler rounding algorithm with a weaker $O(1)$-factor competitive ratio in~\Cref{sec:rounding1}, and defer the slightly more involved $(2+\epsilon)$-competitive algorithm to~\Cref{sec:2-approx-rounding}.
 


\subsection{Producing Fractional Solutions Online using Generalized Flow Instance}
In the first step, we reduce the online load balancing problem to the online generalized network flow problem. We use $G' = (V', E')$ to denote the digraph for the network flow problem. Initially, $V' = M \cup \{\tau\}$, where $\tau$ is the sink.  There is an edge $i\tau \in E'$ with $\mu_{i\tau} = (1+\epsilon)T^*,  c_{i\tau} = 0$ and $\gamma_{i\tau} = 1$. For each arriving job $j$, we add $j$ to $V'$ and the source set.  For every $ij \in E$, we add a directed edge $ji$ to $E'$ with $\mu_{ji} = \infty, c_{ji} = 1$ and $\gamma_{ji} = p_{ij}$.   Below we let $E$ be the final set of edges between jobs and machines. 
\begin{restatable}{lemma}{lemmalbmaintainfractional}
	\label{lemma:lb-maintain-fractional}
	We can maintain fractional solutions $(\bfx^{(t)})_{t \in [n]}$ online such that the following conditions hold:
	\begin{itemize}
		\item For every $t \in [n]$, $\bfx^{(t)} \in [0, 1]^E$ is constructed at time $t$, and is a fractional solution for jobs $[t]$ of makespan at most $(1+\epsilon) T^*$:  $\sum_{i}x^{(t)}_{ij} = 1$ for all jobs $j \in [t]$, $x^{(t)}_{ij} = 0$ for every $ij \in E$ with $j > t$, and  $\sum_{j} p_{ij} x^{(t)}_{ij} \leq (1+\epsilon) T^*$
		for all $i \in M$. \footnote{Technically, we do not know the edges incident to jobs arrived after $t$. However the $x^{(t)}_{ij}$ values for these edges are $0$. Letting $\bfx^{(t)} \in [0, 1]^E$ is only for the sake of notational convenience.}
		\item The total fractional recourse until any time $t$ is bounded, i.e., $\sum_{t'=1}^{t} \sum_{i} |\bfx^{(t')} - \bfx^{(t'-1)}|  \leq O\big(\frac1\epsilon\big) \cdot t$.
	\end{itemize}
\end{restatable}

\subsection{Online Bipartite Matching with Vertex Updates On the Right Side} \label{subsec:b-matching-changing}
Both our rounding algorithms work by carefully constructing an online instance of a matching-type problem, \emph{based on the load balancing fractional solution}, and then using good online algorithms for this instance to derive our final assignment. We therefore describe this variant of the online bipartite matching problem where \emph{vertices on the right side may also change}:

There is a bipartite graph $H = (L \uplus R, E_H)$ with $L = E_H = \emptyset$ initially. $H$ is changed dynamically by the following three types of updates: (i) a new vertex $u$ is added to $L$, along with its incident edges, (ii) a new vertex $v$ is added to $R$, along with its incident edges, and (iii) a vertex $v \in R$ is removed. Our goal is to maintain a matching covering $L$, given the promise that one exists at any time. The algorithm has an amortized recourse of $\beta$, if the number of times we re-assign vertices in $L$ is at most $\beta$ times the number of updates. 

We now show that good algorithms exist if the graph $H$ has sufficient expansion at all times.
	\begin{theorem}
		\label{thm:online-matching-changing}
		Consider an online bipartite-matching problem with the three types of updates specified above.  Let $\alpha > 1$ be a real number such that at any time, we have $|N_H(A)| \geq \alpha|A|$ for every $A \subseteq L$.  Then there is an algorithm for the instance with $O(\log_{\alpha} n)$-amortized recourse, where $n$ is the size of $|L|$ at the end.
	\end{theorem}


	The algorithm is simple: We maintain a matching $F$ in $H$ covering $L$. Whenever a new vertex in $L$ arrives, we match it by changing $F$ using the shortest augmenting path. Whenever some vertex $v \in R$ is removed, we remove the edge in $F$ incident to $v$ if it exists. Then for the unmatched vertex, we match it again using the shortest augmenting path. Due to the $\alpha$-factor slack in the neighborhood, we can use an expansion argument to show that there will always exist a short augmenting path if there is a unmatched vertex in $L$ (\Cref{lemma:shortest-augmenting}).
	
	\begin{restatable}{lemma}{lemmashortestaugmenting}
    	\label{lemma:shortest-augmenting}
    	Let $H = (L \uplus R, E_H)$ be a bipartite graph. Let $\alpha > 1$ be a real number such that $|N_H(A)| \geq \alpha|A|$ for every $A \subseteq L$.  
    	Let $F \subseteq E_H$ be a partial matching where not all vertices in $L$ are matched. Then there is an augmenting path of length at most $2D+1$ w.r.t $F$, where $D = \floor{\log_{\alpha}|L|} + 1 > \log_{\alpha}|L|$. 
    \end{restatable}
    
This finishes the proof of Theorem~\ref{thm:online-matching-changing}, as the number of times we re-assign vertices in $L$ is at most the total length of augmentation paths over time, which is small by~\Cref{lemma:shortest-augmenting}.  

\noindent {\bf Online $b$-Matching with Capacity Updates.}~\Cref{thm:online-matching-changing} can also be generalized to solve the variant where the vertices on the right-side have capacities $b(v)$ which can change over time. The goal is to maintain matchings respecting these capacities, and an algorithm is said to have $\beta$-amortized recourse if the total number of re-assignments is at most $\beta$ times the total change to the capacities $\sum_{v \in R} \sum_{t} | b^{t}(v) - b^{t-1}(v) |$. To see how, note that we can simply have $b(v)$ copies of each vertex $v$ on the right side to reduce it to the online bipartite problem stated above. 


\subsection{Simple Online Rounding via Online $b$-Matching with Right Vertex Updates } \label{sec:rounding1}
We now show how to use the fractional solutions $(\bfx^{(t)})_{t \in [n]}$ to construct the online $b$-matching instance with capacity updates, and apply Theorem~\ref{thm:online-matching-changing} to construct the assignment. 

We use $H = (L \uplus R, E_H)$ and $b:R \to \Z_{\geq 0}$ to denote the graph for the $b$-matching instance and the capacity values of $R$. (We let $b$ be a function instead of a vector due to the heavy notations for vertices in $R$.) The algorithm is formally defined in Algorithm~\ref{alg:olb}. 
Left vertices correspond to jobs. Right vertex $v_{ik} \in R$ associated with machine $i$ is only catering to jobs $j$ with $p_{ij} \approx T^*/2^k$. Indeed, let $f_{ik}$ denote the total fractional allocation of such jobs to machine $i$ in the solution $(\bfx^{(t)})$. 
	Then we ideally want to define the capacity of $v_{ik}$ to be $b({v_{ik}}) = \lceil 2 f_{ik} \rceil$, where the factor of $2$ would guarantee the condition in Theorem~\ref{thm:online-matching-changing} holds with $\alpha = 2$. However, this could lead to too many changes to ${v_{ik}}$'s capacity over time, e.g., if the fractional allocation keeps fluctuating around a half-integral value.  To overcome this, we introduce a random offset and bound the expected change in capacities in terms of the total fractional change. We can also derandomize this using ideas similar to our other rounding algorithm in~\Cref{sec:2-approx-rounding}.




\begin{algorithm}
	\caption{Online Rounding of Fractional Solutions $\bfx^{(1)}, \bfx^{(2)}, \cdots, \bfx^{(n)}$}
	\label{alg:olb}
	\begin{algorithmic}[1]
		\State Let $L \gets \emptyset, R \gets \{v_{ik}: i \in M, k \in [0, K]\}$ for $K = \lceil2 \log n\rceil$, $H \gets (L \uplus R, E_H = \emptyset)$
		\State Choose a random threshold $\rho \in [0,1)$.
		\For{$t = 1, 2, \ldots, n$}
		\State Add $t$ to $L$, its incident edges incident to $E_H$ so that $N_H(t) = \{v_{ik} \, : it\in E, \, T^* / 2^{k+1} < p_{it}  \leq T^* / 2^k \}$.
		\State Update $J_{ik} \gets \{j\in[t] : T^* / 2^{k+1} < p_{ij}  \leq T^* / 2^k\}$, $f_{ik} \gets \sum_{j \in J_{ik}} x^{(t)}_{ij}$ and $b(v_{ik}) \gets \left\lceil 2 f_{ik}  + \rho \right\rceil$ for all $i,k$ \label{albolb:cap}
		\State Follow the algorithm in Theorem~\ref{thm:online-matching-changing} to obtain a new $b$-matching in $H$.
		\State Update $\sigma:[t] \to M$ by setting $\sigma(j) = i$ if $j$ is assigned to $v_{ik}$ for some $k$ in the $b$-matching \label{algolb:assgt}
		\EndFor
	\end{algorithmic}
\end{algorithm}  

Throughout this section, we use the superscript $(t)$ over a notion to denote the value of the parameter at the end of time $t$ during Algorithm~\ref{alg:olb}.

\paragraph{Analysis of $O(1)$-Competitive Ratio and $O(\log n)$-Amortized Recourse}
\begin{lemma}
	\label{lemma:slack}
	At any time of the algorithm, the $b$-matching instance $(L \cup R, E_H)$ satisfies $b(N_H(A)) \geq 2|A|$ for every $A \subseteq L$.	
\end{lemma}

\begin{proof}
	The lemma holds since $\bfx^{(t)}$ gives a fractional matching between $[t]$ and $R$ where every $j \in [t]$ is matched to an extent of $1$ and every $v_{ik}$ is matched to an extent of $f^{(t)}_{ik}$, and $b^{(t)}(v_{ik}) \geq 2 f^{(t)}_{ik}$ is an integer.
\end{proof}

By Theorem~\ref{thm:online-matching-changing}, the algorithm achieves an $O(\log n)$-amortized recourse for the online matching instance.

\begin{claim} \label{cl:tau-expct}
	Let $f$ and $f'$ be any two real values, and let $\rho$ be a random variable sampled uniformly from $[0,1)$. Then $\mathbf{E}_{\rho} \left[ \left|  \lceil f + \rho \rceil -  \lceil f' + \rho \rceil \right| \right] = \left|f - f'\right|$.
\end{claim}

\begin{proof}
	The proof is straightforward. Indeed, suppose without loss of generality, $f \geq f'$. Then, $\lceil f + \rho \rceil -  \lceil f' + \rho \rceil$ is precisely the number of integral values which lie in $f + \rho$ and $ f' + \rho$. But because $\rho$ is a uniformly random value in $[0,1)$, it is easy to see that this number is precisely $|f - f'|$. 
\end{proof}

\begin{lemma} \label{lem:capchange}
	The total expected change in capacities $\mathbf{E}_{\rho} \left[ \sum_{t' =1}^{t} \sum_{ik} \left| b^{(t')}(v_{ik}) - b^{(t'-1)}(v_{ik}) \right| \right]$ by time $t$, is at most $\sum_{t' =1}^{t}\left| \bfx^{(t')} - \bfx^{(t'-1)} \right|$.
\end{lemma}

\begin{proof}
	The proof is a simple argument using linearity of expectation and the triangle inequality. For any $i$ and $k$, we use $J_{ik}$ to denote the final set $J_{ik}$, i.e, $\{j \in [n] : T^* / 2^{k+1} < p_{ij}  \leq T^* / 2^k\}$
	\begin{align*}
		&\quad \mathbf{E}_{\rho} \left[ \sum_{t' =1}^{t} \sum_{ik} \left| b^{(t')}(v_{ik}) - b^{(t'-1)}(v_{ik}) \right| \right] \ =\  \sum_{t' =1}^{t} \sum_{ik}  \mathbf{E}_{\rho} \left[ \left| b^{(t')}(v_{ik}) - b^{(t'-1)}(v_{ik}) \right| \right] \\
		&= 
		\sum_{t' =1}^{t} \sum_{ik} \mathbf{E}_{\rho} \left[ \left| \lceil 2f^{(t')}_{ik} + \rho \rceil - \lceil 2f^{(t'-1)}_{ik} + \rho \rceil \right| \right] 
		\ = \ \sum_{t' =1}^{t} \sum_{ik} 2\left| f^{(t')}_{ik} - f^{(t'-1)}_{ik} \right| 
		\ \leq \ 2\sum_{t' =1}^{t} \left| \bfx^{(t')} - \bfx^{(t'-1)} \right|.
	\end{align*}
	Above, the third equality follows from~\Cref{cl:tau-expct}, and the next inequality follows from the triangle inequality, and by noting that the various $J_{ik}$ sets over all $k$ for a fixed $i$ form a disjoint partition of all the jobs.
\end{proof}

\begin{lemma} \label{lem:lbcost}
	After arrival of the first $j$ jobs, the load on any machine $i$ due to our assignment $\sigma_j$ in~\Cref{algolb:assgt} is at most $O(1) T^*$.
\end{lemma}

\begin{proof}
	Consider a fixed machine $i$. Below, we use $\pi^{(t)}_{jik}$ to indicate if $j$ is assigned to $v_{ik}$ in the solution for the $b$-matching instance at time $t$.  The total load $i$ receives at time $t$ is then
	\begin{align*}
		\sum_{k = 0}^{K} \sum_{j \in J_{ik}} \pi^{(t)}_{jik} p_{ij} & \leq  \sum_{k = 0}^{K} \frac{T^*}{2^k} \sum_{j \in J_{ik}} \pi^{(t)}_{jik} 
		\ \leq\ \sum_{k = 0}^{K} \frac{T^*}{2^k} \cdot b^{(t)}(v_{ik})
		\ \leq\  \sum_{k = 0}^{K} \frac{T^*}{2^k} \left(2\sum_{j \in J_{ik}} x^{(t)}_{ij} + 2 \right) \\
		& \leq O(T^*) + 2 \sum_{k = 0}^{K} \frac{T^*}{2^k} \sum_{j \in J_{ik}} x^{(t)}_{ij} 
		\ \leq\ O(T^*) + 4 \sum_{k = 0}^{K} \sum_{j \in J_{ik}} x^{(t)}_{ij} p_{ij} 
		\ \leq\ O(T^*) + 4 \sum_{j \in [t]} x^{(t)}_{ij} p_{ij} 
		\ \leq\ O(T^*).
	\end{align*}
\end{proof}

	The competitive ratio follows from~\Cref{lem:lbcost}. As for the recourse, note that the total recourse made by the load balancing algorithm is at most the number of reassignments made by the $b$-matching algorithm. The recourse of the latter by time $t$ is at most $O(\log n)$ times $t + \sum_{t'=1}^{t} \sum_{ik} | b^{(t')}(v_{ik}) - b^{(t'-1)}(v_{ik}) |$ by Theorem~\ref{thm:online-matching-changing}. From~\Cref{lem:capchange}, the expected value of the summation is at most $\sum_{t' =1}^{t} \left| \bfx^{(t')} - \bfx^{(t'-1)} \right|$. By Lemma~\ref{lemma:lb-maintain-fractional}, this is at most $O(t)$. This shows the the amortized recourse of the algorithm is $O(\log n)$. 

	\section{An $O\left(\frac{\log \log n}{\log \log \log n}\right)$-Competitive Algorithm with $O(1)$-Amortized Recourse} \label{sec:const-recourse}
	    In this section, we describe the $O\left(\frac{\log \log n}{\log \log \log n}\right)$-competitive algorithm for unrelated machine load balancing (\olb) with $O(1)$-amortized recourse. Notice that we can lose an $O(1)$-factor in the competitive ratio and thus we can simply fix $\epsilon = 1$. Using Lemma~\ref{lemma:lb-maintain-fractional}, we construct a sequence $\bfx^{(1)}, \bfx^{(2)}, \cdots, \bfx^{(n)}$ of fractional solutions online, each $\bfx^{(t)}$ being a fractional schedule for jobs $[t]$. We  assume the makespans of the fractional solutions are $T^*$ by scaling up $T^*$. For every $t \in [n]$, we have $\sum_{t'=1}^t |\bfx^{(t')} - \bfx^{(t'-1)}| = O(1) \cdot t$.

		We prove the following theorem in this section, which in turn proves Theorem~\ref{thm:unrelated-constant-approx}.
		\begin{theorem}
			\label{thm:O(1)-recourse-main}
			There is a randomized online rounding algorithm that with high probability produces solutions of makespan $O\left(\frac{\log\log n}{\log \log \log n}\right)\cdot T^*$ and incurs recourse $O(1)\cdot t$ by time $t$, for every $t \in [n]$.
		\end{theorem}
	
		For convenience, we add a dummy machine $i_\bot$ to $M$, and assign jobs that have not arrived yet to $i_\bot$ both in input fractional schedules and output integral schedules. That, is we assume $x^{(t)}_{i_\bot j} = 1$ for every $j > t$.  After this transformation, every $\bfx^{(t)}$ is a schedule for the whole job set $[n]$.  We assume all jobs have processing time $0$ on the dummy machine; this will not create an issue as our algorithm will never assign a job $j$ to a machine $i$ at time $t$ if $x^{(t)}_{ij} = 0$.

	\subsection{Main Ideas}
		Our algorithm uses the framework of the $O\left(\frac{\log \log m}{\log \log \log m}\right)$-competitive online rounding algorithm of \cite{LX21}. (Recall that $m$ is the number of machines.) In their setting, the fractional assignment of a job $j$ never changes after its arrival. As a result, their algorithm does not need to incur a recourse.
				
		We give a high-level overview of the rounding algorithm in \cite{LX21}. They describe the algorithm in the offline setting, and one can easily make it online.  We say an edge $ij \in E$ is big if $p_{ij} > \frac{T^*}{\log m}$, and small otherwise. A job $j$ is big if at least $1/2$ fraction of the job is  assigned via big edges in the fractional solution, and small otherwise. For big (small) jobs $j$, we only consider its big (small) edges. Small jobs can be assigned by independent rounding; with high probability they incur only an $O(1)\cdot T^*$ load on every machine. Thus the bulk of the algorithm is for the assignment of big jobs, which is done in three steps:
		\begin{itemize}
			\item Step (b1): The algorithm does an initial rounding to make the support of $\bfx$ sparse: If some $x_{ij}$ for a big job $j$ has $x_{ij} \in (0, \frac1{\log m})$, then it rounds $x_{ij}$ to $0$ or $\frac1{\log m}$ randomly, preserving the expectation of $x_{ij}$.  After this step, every such $x_{ij}$ is either $0$ or at least $\frac1{\log m}$.  This guarantees that the support graph for $\bfx$ restricted to big jobs have degree $O(\log^2 m)$.
			\item Step (b2): The algorithm attempts to assign every big job $j$ to a machine $i$ randomly, using the new $x_{ij}$ values as probabilities. The assignment fails if the target machine is overloaded. 
			\item Step (b3): The crucial theorem proved in \cite{LX21} is that the following event happens with high probability: In the sub-graph of the support of $\bfx$ induced by the failed jobs and all machines, every connected component has size at most $\text{poly}\log m$. This allows the algorithm to apply a deterministic $O\left(\frac{\log \log m}{\log \log \log m}\right)$-competitive online rounding procedure for each component.
		\end{itemize} \medskip
	
		We use a similar framework as that of \cite{LX21}, with the following main differences. First, we generate a set of global random seeds that are used in our randomized rounding procedure for each time step. They will correlate schedules at different time steps. Second, in step (b3), we can run the simple offline 2-approximation algorithm for each connected component, as recourse is allowed in our setting. Finally a small difference is that our competitive ratio is $O\left(\frac{\log \log n}{\log \log \log n}\right)$ as we need to apply union bounds over $n$ time steps.
	
	\subsection{Description of Algorithm and Proof of Competitive Ratio}
	We now formally describe our  algorithm. We say an edge $ij \in E$ is big if $p_{ij} \geq \frac{T^*}{{\log n}}$, and small otherwise. For every $j \in [n]$, we let $M^\rmbig_j$ ($M^\rmsmall_j$ resp.) be the set of machines $i \in N(j)$ with $ij$ being big (small resp.). 
	
	\paragraph{Generating Global Random Seeds} We choose a threshold $\beta \in [\frac12, \frac34]$ uniformly at random, that will be used to define big and small jobs. 
	\begin{definition}
		Given $\beta$ and a fractional solution $\bfx \in [0, 1]^E$, we say a job $j$ scheduled in $\bfx$ 
		 is big if $\sum_{i \in M^\rmbig_j}x_{ij} > \beta$ and small otherwise. So, if $j$ is small, then $\sum_{i \in M^\rmsmall_j}x_{ij} \geq 1 - \beta$. Let $J^\rmbig$ and $J^\rmsmall$ be the sets of big and small jobs respectively. 
	\end{definition}
	
	 When a job $j$ arrives, for every big edge $ij \in E$, we independently  choose a threshold $\delta_{ij} \in [0, 1/{\log n}]$ uniformly at random; this will be used in the initial rounding step (step (b1)) for big jobs.  We also generate an infinite sequence of pairs $(h^j_1, \theta^j_1), (h^j_2, \theta^j_2), (h^j_3, \theta^j_3), \cdots$, where for every $o \in \Z_{>0}$, $h^j_o$ is a random machine in $N(j)$, and $\theta^j_o$ is a random real number in $[0, 1]$; all the parameters are independently generated.  The sequence will serve as the random seeds for the acceptance-rejection sampling method, to assign small jobs, and to assign big jobs in step (b2). \medskip
	
	Once we generated the global random seeds, it is convenient to describe the algorithm \emph{in the offline setting}: For every time step $t$, we round the fractional solution $\bfx^{(t)}$ to obtain an integral assignment, using the global random seeds.  A recourse occurs when the assignment of a job $j$ at time $t$ is different from that at time $t-1$.  %
	%
	So till the end of this section we focus on a time step $t$, and the fractional solution $\bfx := \bfx^{(t)} \in [0, 1]^E$. Big and small jobs are defined w.r.t the global seed $\beta$ and this fractional solution $\bfx$.
	
	\paragraph{Step (s): Assigning Small Jobs} For every $j \in J^\rmsmall$, we find the smallest $o \in \Z_{>0}$ such that $h^j_o \in M^\rmsmall_j$ and $\theta^j_o \leq x_{h^j_oj}$, and we assign $j$ to $h^j_o$. This is equivalent to the following procedure. We draw a histogram for job $j$, where there is a bar of height $x_{ij}$ for every $i \in M^\rmsmall_j$, and a bar of height $0$ for every $i \in M^\rmbig_j$. Then each $(h^j_o, \theta^j_o)$ denotes a random point in $N(j) \times [0, 1]$; we accept a point if it falls into a bar in the histogram. Suppose  $(h^j_o, \theta^j_o)$ is the first point we accept; we then assign $j$ to $h^j_o$. Therefore, conditioned on $j \in J^\rmsmall$, the probability that $j$ is assigned to a machine $i \in J^\rmsmall_j$ is proportional to $x_{ij}$.  Moreover, by the definition of the small jobs, the probability is between $x_{ij}$ and $\frac{1}{1-\beta}x_{ij} \leq 4x_{ij}$.
	\begin{lemma}
		With probability at least $1 - \frac{1}{n^3}$, every machine $i$ gets a load of at most $O(1)\cdot T^*$ from step (s).
	\end{lemma}
	\begin{proof}
		This follows from the Chernoff bound and union bound, as the assignments of different small jobs $j$ are independent. 
	\end{proof}
	
	It remains to assign big jobs; we do this in three steps.
	\paragraph{Step (b1): Initial Rounding of $\bfx$ Values}
		For a job $j \in J^\rmbig$, and every $i \in M^\rmbig_j$, we define $x'_{ij}$ as follows:
		\begin{align*}
			x'_{ij} := \begin{cases}
				x_{ij} & \text{if } x_{ij} \geq \frac1{\log n}\\
				0 & \text{if } x_{ij} < \frac1{\log n} \text{ and } \delta_{ij} > x_{ij}\\
				\frac1{\log n} & \text{if } x_{ij} < \frac1{\log n} \text{ and } \delta_{ij} \leq x_{ij}
			\end{cases}
		\end{align*}
		For other $ij$ pairs we let $x'_{ij} = 0$.  By the way we generate $\delta_{ij}$ values, we have that $\E[x'_{ij}] = x_{ij}$ conditioned on $j \in J^\rmbig$ and $i \in M^\rmbig_j$. 
		\begin{lemma}
			\label{lemma:x'-good}
			With probability at least $1 - 1/n^2$, the following events happen:
			\begin{itemize}
				\item For every $j \in J^\rmbig$, we have $\sum_{i \in M^\rmbig_j}x'_{ij} \geq \frac1{10}$.
				\item For every $i \in M$, we have $\sum_{j \in N(i)}p_{ij}  x'_{ij} \leq 10 T^*$.
			\end{itemize}
		\end{lemma}
		\begin{proof}
			The rounding procedure for jobs in $J^\rmbig$ are independent. Again we can apply Chernoff bound and Union bound to prove the lemma. 
		\end{proof}
	
		From now on, we assume the above events happen. Then, every machine $i$ is incident to at most $O(\log^2n)$ jobs in the support of $\bfx'$, as each positive $x'_{ij}$ has $p_{ij} >\frac{T^*}{{\log n}}$ and $x'_{ij} \geq \frac1{\log n}$. 
		
		\paragraph{Step (b2): Attempts to Assign Big Jobs} For every $j \in J^\rmbig$, we find the smallest $o \in \Z_{>0}$ such that $h^j_o \in M^\rmbig_j$ and $\theta^j_o \leq x_{h^j_oj}$, and temporarily assign $j$ to $i := h^j_o$.  We \emph{mark} a machine $i$ if it gets a load of more than $\frac{c\log\log n}{\log\log\log n} \cdot T^*$ in this step (that is, we do not count the load assigned from step (s)), for a sufficiently large constant $c$.  All the jobs in $J^\rmbig$ assigned to marked machines are said to be \emph{failed}, and we undo these assignments.\footnote{There is a slight difference between our algorithm and that of \cite{LX21} in this step. In \cite{LX21}, for a marked machine $i$, the jobs assigned to $i$ before it is marked are not failed. This is needed for their no-recourse setting. }  Let $J^\fail$ be the set of failed jobs.  If some $j \in J^\rmbig$ is assigned to an unmarked machine, i.e., $j \in J^\rmbig\setminus J^\fail$, then the assignment of $j$ is successful and final for time step $t$.

		\paragraph{Step (b3) : Assign Failed Jobs Component by Component}  We let $G' = (M \uplus J^\rmbig, E')$ be the support bipartite graph for $\bfx'$: for some $j \in J^\rmbig, i \in M$, we have $ij \in E'$ if and only if $x'_{ij} > 0$.  Then $G'[M \cup J^\fail]$ is the sub-graph of $G'$ induced by $M \cup J^\fail$.
		
		For any connected component $C$ in $G'[M \cup J^\fail]$. We run the deterministic $2$-approximation algorithm~\cite{lenstra1990approximation} to obtain an assignment of failed jobs in $C$ to $M$. We can guarantee that any machine is assigned a load of at most $200 T^*$ in this step.  This holds because $\bfx'$ is an approximate LP solution, where by Lemma~\ref{lemma:x'-good} the two constraints are each violated by a factor of $10$.
		
		This finishes the description of the algorithm. With high probability, the algorithm maintains a schedule of makespan at most $O\left(\frac{\log \log n}{\log \log \log n}\right)\cdot T^*$ at any time: Small jobs, successful and failed big jobs respectively incur a load of $O(1)\cdot T^*, O\left(\frac{\log \log n}{\log \log \log n}\right)\cdot T^*$ and $O(1)\cdot T^*$ on each machine.
		

		
	\subsection{Bounding the Recourse}
	 In this section, we bound the recourse of the algorithm case by case.  Below we fix a time step $t \geq 1$, and use $\bfx^\circ$ and $\bfx$ to denote the fractional solution at time $t-1$ and $t$ respectively.  For every job $j$, we let $\bfx_j$ denote the vector $(x_{ij})_{i \in N(j)}$, i.e., the fractional assignment of $j$ at time $t$. Define $\bfx^\circ_j$ similarly for time $t-1$.   Define $\bfx'^\circ, \bfx', \bfx'_j, \bfx'^\circ_j$ similarly for the vector $\bfx' \in [0, 1]^E$ obtained in step (b2). In all the proofs, we assume $c$ is a sufficiently large constant. 
	
	\paragraph{Type-1 recourse: recourse from switches between small and big jobs} We say $j$ incurs a type-1 recourse at time $t$ if $j$ is small at time $t-1$, and big at time $t$, or the other way around.
	\begin{lemma}
		Let $j \in [n]$.  In expectation over the randomness of $\beta$, the probability that $j$ incurs a type-1 recourse at time $t$ is at most $2\cdot \big| {\bfx}_j -  {\bfx}^\circ_j \big|$.
	\end{lemma} 
	\begin{proof}
		Let $p^\circ_\rmbig$ be the fraction of job $j$ assigned via big edges, in the fractional solution $\bfx^\circ$; define $p_\rmbig$ similarly for $\bfx$.  Then we have $|p_\rmbig - p^\circ_\rmbig| \leq \frac12\big| {\bfx}_j -  {\bfx}^\circ_j \big|$. If there is a switch between $j$ being small and $j$ being big, then $\beta$ must be between $p_\rmbig$ and $p^\circ_\rmbig$.  This happens with probability at most $4 \cdot |p_\rmbig - p^\circ_\rmbig| \leq 2\cdot \big| {\bfx}_j -  {\bfx}^\circ_j \big|$ since $\beta$ is uniformly chosen from $\left[\frac12,\frac34\right]$.
	\end{proof}
	
	\paragraph{Type-2 recourse: recourse from small jobs} We say a recourse incurred by $j$ at time $t$ is of type-2 if $j$ is small at both time steps $t-1$ and $t$, but assigned to different machines in the two time steps.
	\begin{lemma}
		For every $j \in [n]$, the probability that $j$ incurs a type-2 recourse at time $t$  is at most $4\cdot |\bfx_j - \bfx^\circ_j|$.
	\end{lemma}
	\begin{proof}
		We fix a $\beta$ for which $j$ is small in both time $t-1$ and $t$. Consider the point $(h^j_1, \theta^j_1)$ from the random sequence for $j$. We say the point  is good if it is accepted in both time step $t - 1$ and $t$; we say the point is bad if is accepted in exactly one of the two time steps. The point is neutral if it is not accepted in either time step.  The probability that $j$ incurs a type-2 recourse at time $t$ is at most $\frac{\Pr[(h^j_1, \theta^j_1)\text{ is bad}]}{\Pr[(h^j_1, \theta^j_1)\text{ is good or bad}]} \leq \frac{\big|\bfx_j - \bfx^\circ_j\big|/|N(j)|}{1/(4|N(j)|)}  = 4 \cdot \big|\bfx_j - \bfx^\circ_j\big|$. De-conditioning on $\beta$ gives the lemma. 
	\end{proof}

	\paragraph{Type-3 recourse: recourse from different target machines in step (b2)} We say a recourse incurred by $j$ at time $t$ is of type-3 if $j$ is big at both time steps $t-1$ and $t$, but it is temporarily assigned to different machines in step (b2) of the two time steps. 
	\begin{lemma}
		For every $j \in [n]$, the probability that $j$ incurs a type-3 recourse at time $t$ is at most $10\cdot |\bfx_j - \bfx^\circ_j|$. 
	\end{lemma}
	\begin{proof}
		First we condition on the value of $\beta$ such that $j$ is big at both time $t-1$ and $t$. Conditioned on $\delta_{ij}$ values, the probability that $j$ incurs a type-3 recourse at time $t$ is at most $10\cdot |\bfx'_j - \bfx'^\circ_j|$, using the same argument from the previous lemma.    De-conditioning on $\delta_{ij}$ values, the probability is at most $10 \cdot |\bfx_j - \bfx^\circ_j|$; this holds since $E_{\delta_{ij}}[|x'_{ij} - x'^\circ_{ij}|] = |x_{ij} - x^\circ_{ij}|$ for some $i \in M^\rmbig_j$.   Finally, de-conditioning on $\beta$ gives the lemma. 
	\end{proof}
	
	\paragraph{Type-4 recourse: recourse from failed jobs}
	We say a recourse incurred by a job $j$ at time $t$ is of type-4, if it is not of type-1, 2 or 3.  In this case, $j$ is big in both time steps $t-1$ and $t$, it is temporarily assigned to the same machine in step (b2), and it fails in at least one of the two time steps.  Moreover, the connected component in $G'[M \cup J^\fail]$ containing $j$ is not the same in the two time steps (we assume the condition holds if $j$ does not fail in the other time step). 
	
	 Let $\calC$ be the set of connected components in $G'[M \cup J^\fail]$ in time step $t$; define $\calC^\circ$ similarly for time step $t-1$.  So we need  to count the number of jobs in the components in $\calC^\circ \,\triangle\, \calC$, where $\calC^\circ \,\triangle\, \calC$ denotes the symmetric difference between $\calC^\circ$ and $\calC$, i.e., the set of components appearing exactly one of $\calC^\circ$ and $\calC$.  Due to the symmetry, it suffices to consider the components in $\calC^\circ \setminus \calC$. 
	
	The key theorem we use is that any $C \in \calC^\circ$ is small with high probability. This is proved in \cite{LX21} and used to bound their competitive ratio; while in our case, we apply the theorem to bound the recourse. The crucial properties we use to prove this are the facts that  $G'$ has degree $O(\log^2(n))$, and that every machine is marked with probability $1/\poly\log(n)$, with a sufficiently large exponent in the $\poly\log(n)$ factor. We omit its proof here as it is identical to that in \cite{LX21}, except for minor differences in parameters.
	\begin{theorem}[Lemma 4.5 in \cite{LX21}]
		\label{thm:small-component}
		With probability at least $1 - \frac{1}{n^3}$, every $C \in \calC^\circ$ contains at most $O(\log^7 n)$ machines.  
	\end{theorem}

	With the theorem, we can now focus on bounding $|\calC^\circ \setminus \calC|$.  The change of components from $\calC^\circ$ to $\calC$ are caused by the following two types of events:
	\begin{itemize}
		\item Some job $j$ has changed its target machine in step (b2) from time $t-1$ to $t$, and in at least one of the two time steps, the target machine is marked; we also say this event happens if $j$ is small in the other time step. This event may add/remove $j$ to/from $J^\fail$, switch up to 2 machines between marked and unmarked from time $t-1$ to $t$. We say this is a type-a event incurred by job $j$. 
		\item For some $j$ that fails in both time step $t-1$ and $t$, and some $i \in M^\rmbig_j$, we have $x'_{ij} \neq x'^\circ_{ij}$. We say this is a type-b event incurred by the pair $ij$.  This will add/remove the edge $ij$ to/from the graph $G'$. 
	\end{itemize}
	Notice that if no events of the two types occur, then $\calC^\circ = \calC$, as marked machines, failed jobs, and $x'$ values incident to failed jobs do not change from time $t-1$ to time $t$. 
	
	\begin{lemma}
		Fix a job $j \in [n]$. The probability that $j$ incurs a type-a event is at most $O\left(\frac{1}{\log^{14}n}\right) \cdot |\bfx_j - \bfx^\circ_j|$.
	\end{lemma}
	\begin{proof}
		The probability that $j$ is assigned to two different target machines in step (b2) in time steps $t-1$ and $t$ is $O(1) \cdot |\bfx_j - \bfx^\circ_j|$. Conditioned on the event, the probability that the target machine is marked in the correspondent time step is at most $O\left(\frac{1}{\log^{14}n}\right)$, when $c$ is big enough. 
	\end{proof}

	\begin{lemma}
		Fix a big edge $ij \in E$. The probability that $ij$ incurs a type-b event is at most $O\left(\frac{1}{\log^{14}n}\right) \cdot |x_{ij} - x^\circ_{ij}|$.
	\end{lemma}
	\begin{proof}
		With probability at most $O(\log n) \cdot |x_{ij} - x^\circ_{ij}|$, the $x'_{ij}$ value is different in time step $t-1$ and $t$, and $j$ is big in both time steps.  Under the condition that this happens, the probability that $j$ fails at time step $t$ is at most $O\left(\frac{1}{\log^{15}n}\right)$. The lemma then follows.
	\end{proof}

	Each type-a event can change at most $O(\log^2n)$ components in $\calC^\circ$: it may change the marking status of two machines, and a machine is incident to $O(\log^2n)$ jobs in the support of $x'^\circ$.   Each type-b event can change at most at most 2 components in $\calC^\circ$.  Therefore, the expected number of components in $\calC^\circ \setminus \calC$ is at most $\sum_{j \in [n]}O\left(\frac1{\log^{14}n}\right)\big|\bfx_j - \bfx^\circ_j\big| \cdot O(\log^2n) = O\left(\frac{1}{\log^{12}n}\right)\cdot\sum_{j \in [n]}|\bfx_j - \bfx^\circ_j|$. By Theorem~\ref{thm:small-component}, in expectation, the type-4 recourse is at most $O\left(\frac{1}{\log^{12}n}\right) \sum_{j \in [n]}|\bfx_j - \bfx^\circ_j|  \cdot O(\log^7 n) \cdot O(\log^2 n)= O(\frac{1}{\log^3n})\sum_{j \in [n]}|\bfx_j - \bfx^\circ_j|$.

	Therefore, we have proved that the expected recourse at time $t$ is at most $O(1) \sum_{j \in [n]}|\bfx_j - \bfx^\circ_j| = |\bfx - \bfx^\circ|$. Summing up the bound over all time steps $t'$ from $1$ to $t$, we obtain that the recourse by time $t$ is at most $\sum_{t' = 1}^t |\bfx^{(t')} - \bfx^{(t'-1)}| = O(1) \cdot t$. This finishes the proof of Theorem~\ref{thm:O(1)-recourse-main}.






{\small
\bibliographystyle{alpha}
\bibliography{bib,refs}
}

\appendix
\section{Missing Proofs}
\label{appendix:missing-proofs}

\thmbmat*
\begin{proof}
    We reduce \obmp to the online generalized network flow problem (indeed, the ordinary network flow problem suffices as we all gains are defined as $1$).  We maintain a digraph $G' = (L \cup R \cup \{\tau\}, E')$. We have edges from $L$ to $R$ in $G'$ that are the same as those in $G$ (except that the edges in $G'$ are directed), and edges from each $v \in R$ to $\tau$. An edge $uv \in E'$ with $u \in L, v \in R$ has cost $c_u$, capacity $\infty$ and gain $1$. An edge $v\tau$ for $v \in R$ has cost $0$, capacity $\lceil(1+\epsilon)b_v\rceil$ and gain $1$.  The sources in the network are $L$. The graph $G'$ is constructed online: when a new vertex in $L$ arrives, we add it and its outgoing edges to $G'$.
    
    Theorem~\ref{thm:flow} gives an online algorithm that maintains a network flow where every $u \in L$ sends 1 units of flow. As the edges have gain parameters $1$, $\tau$ receives $|L|$ units of flow. As the algorithm are based on augmenting paths, the flow is integral. Thus the algorithm maintains a matching where every $v \in R$ is matched at most $\lceil(1+\epsilon)b_v\rceil$ times.  The reassignment cost of the algorithm is equal to the cost incurred by the network flow algorithm, which is at most $\frac{1+\epsilon}{\epsilon}$ times the optimum cost $C^*$ of the network flow instance with capacities scaled by $\frac1{1+\epsilon}$. After scaling down the capacities, every edge $v\tau \in E'$ has capacity at least $b_v$. Thus $C^* \leq \sum_{u \in L} c_u$. 
\end{proof}

\lemmanonfractional*
\begin{proof}
	We consider any solution $f$ to the linear system defined by the constraints in~\Cref{def:fractiona-augmenting}, with objective $\cost(f)$. As all edge costs are non-negative, we can assume the optimum solution is achieved when all values in $f$ are bounded. Therefore, it must be achieved at a vertex point  defined by some tight inequalities in the linear system. Let $G' = (V', E')$ be the subgraph of $G^\bfx$ containing the support of $f$, and the vertices incident to these edges. We can assume $G'$ contains $s$ and is connected.  Every $v \in V' \setminus \{s, \tau\}$ has at least one incoming edge and one outgoing edge; $s$ has at least one outgoing edge. As $f$ is a vertex solution, we have $|E'| \leq |V' \setminus \{\tau\}|$. We now consider two cases depending on whether $\tau \in V'$ or not. 
	
	First assume $\tau \in V'$. Then we have $|E'| \leq |V'| - 1$.  Therefore, we have $|E'| = |V'| - 1$ and $G'$ is a path from $s$ to $\tau$ in $G^\bfx$ (case \ref{item:AP-path}).  It remains to consider the case $\tau \notin V'$, which implies $|E'| \leq |V'|$. It can only happen that $|E'| = |V'|$. Since every vertex in $V'$ except for $s$ has at least one incoming and one outgoing edge. Two sub-cases may happen in this case.  It may be that $s$ also has an incoming edge in $G'$, in which case $G'$ is a cycle containing $s$ (case \ref{item:AP-cycle}). It may also be that $s$ does not have an incoming edge in $G'$, but some other vertex in $G'$ has 2 incoming edges. In this case, $G'$ contains a cycle and a path connecting $s$ to the cycle (case \ref{item:AP-union}).  This finishes the proof of the lemma.
\end{proof}

\lemmaheightincreases*

\begin{proof}
	Let $s = s_t$ be the source arrived at time $t$.  Notice that adding $s$ and its outgoing edges to $G$ does not decrease the cost of cheapest augmenting path from $v$ in $G^\bfx$, since $s$ does not have incoming edges. Let $\bfx$ be the flow at the beginning of some iteration of the while loop, at time $t$. Let $f$ be the cheapest augmenting path from $s$ in the residual graph $G^\bfx$, as in Step~\ref{step:find-path}.   Let $\bar \bfx$ be the $\bfx$ obtained at the end of the iteration, i.e., after it is augmented using $f$. Let $f^1$ and $f^2$ be the cheapest augmenting paths from $v$ in the residual graph $G^\bfx$ and $G^{\bar \bfx}$ respectively. Let $C_1 = \cost(f^1)$ and $C_2 = \cost(f^2)$.  It suffices for us to prove $C_2 \geq C_1$. 
	
	
	We first consider the case where $v$ is in the support graph $G'$ of the flow $f$, which falls in one of the three cases in Definition~\ref{def:augmenting}. Towards contradiction we assume $C_2 < C_1$.  By Claim~\ref{claim:break}, we can find an augmenting path $f'$ from $v$ in $G_x$, whose support is a sub-graph of $G'$.  As $f'$ is an augmenting path from $v$ in $G_x$, we have $\cost(f') \geq \cost(f^1) = C_1$ by the choice of $f^1$. 
	
	Let $\theta > 0$ be a sufficiently small real number. We define $f^*:=f - \theta f' + \theta f^2$, extending the domain of the three vectors by adding 0-coordinates if necessary.  When $\theta$ is small enough, all entries in $f^*$ are non-negative. So, the cost of $f^*$ is strictly smaller than that of $f$ as $\cost(f')  \geq C_1 > C_2 = \cost(f^2)$.  Clearly, $f^*$ satisfies the balance constraints: the net flow sent by $s$ in $f^*$ is 1, and the net flow sent by any vertex in $V \setminus \{s, \tau\}$ is 0.  
	
	However, some edges in the domain for $f^2$ may be outside $E^\bfx$. This issue can be handled as follows. Suppose some $e = vu$ with $f^2_e > 0$ is not in $E^\bfx$.  Then it must be the case that $uv \in E^\bfx$ and $f_{uv} > 0$.  Then in $f^*$, we update $f^*_{uv} \gets f^*_{uv} - \theta f^2_{vu}/\gamma_{uv}$ and $f^*_{vu} \gets 0$, and discard the coordinate.  When $\theta > 0$ is sufficiently small, we guarantee that all entries in $f^*$ are non-negative.  Moreover this update operation can only decrease the cost of $f^*$, and thus we still have $\cost(f^*) < \cost(f)$.  As the updated $f^*$ is a fractional augmenting path from $s$ in $G^\bfx$, this leads to a contradiction to the choice of $f$.  \medskip
	
	It remains to consider the case where $v$ is not in the support graph $G'$ of $f$.  We consider first vertex $v'$ along the path $f^2$ that is in $G'$. (In case $f^2$ belongs to case \ref{item:AP-cycle} or \ref{item:AP-union}, the path can be obtained by starting from $v$ and following out-going edges in the support of $f^2$.) Using Claim~\ref{claim:break},  we can break $f^2$ into a flow path $\hat f$ sending $1$ units flow from $v$ to $v'$, and a scaled augmenting path $\hat f'$ from $v'$ in $G^{\bar \bfx}$.  We already proved that the cost of the cheapest augmenting path from $v'$ in $G^\bfx$ is at most that in $G^{\bar \bfx}$.  So, we can replace $\hat f'$ with the cheapest augmenting path from $v$ in $G^\bfx$, scaled by the same factor.  Notice that $\hat f$ is also a path in $G^\bfx$ as $f$ does not involve any vertex before $v'$.  Therefore, we obtained an augmenting path from $v$ in $G^\bfx$ with cost no larger than that of $f^2$.  Therefore, we have $C_2 \geq C_1$. 
	%
	%
\end{proof}

{\renewcommand\footnote[1]{}\lemmalbmaintainfractional*}

\begin{proof}
	Notice that a valid generalized flow for the instance at time $t$ corresponds to a fractional solution $\bfx^{(t)}$ for jobs $[t]$, with makespan $(1+\epsilon)T^*$.  The recourse incurred by the fractional solutions at any time is precisely 2 times the cost made by the online algorithm for the generalized flow problem.  Applying Theorem~\ref{thm:flow}, the recourse by time $t$ can be bounded by $O\left(\frac1\epsilon\right)$ times the cost of the offline generalized flow instance at $t$, with capacities scaled by $\frac1{1+\epsilon}$. As the capacities of incoming edges of $\tau$ are $(1+\epsilon)T^*$, scaling by $\frac1{1+\epsilon}$ reverts the capacities back to $T^*$.  The offline generalized network flow instance at time $t$ has a solution of cost $t$.  So, the recourse incurred by the fractional solution at time $t$ is at most $O\big(\frac1\epsilon\big) \cdot t$. 
\end{proof}

    \lemmashortestaugmenting*

\begin{proof}
	Let $\vec{H}$ be the residual graph of $H$ w.r.t $F$: $\vec H$ is a directed graph over $L \cup R$, for every edge $uv \in E_H$, we have $uv \in \vec H$, and for every $uv \in F$, we have $vu \in \vec H$.   We say a vertex in $L$ is free if it is unmatched in $F$.  For every integer $d \in [0, D]$, define $L^d$ ($R^d$ resp.) to be the set of vertices in $L$ ($R$, resp.) to which there exists a path in $\vec{H}$ of length \emph{at most} $2d$ ($2d+1$, resp.) from a free vertex. So, we have $L^0 \subseteq L^1 \subseteq L^2 \subseteq \cdots \subseteq L^D$ and $R^0 \subseteq R^1 \subseteq R^2 \subseteq \cdots \subseteq R^D$. Also notice that $R^{d} = N_H(L^{d})$ for every $d \in [0, D]$. 
	
	For every $d \in [0, D]$, we have $\alpha|L^{d}| \leq |R^{d}|$ by the condition of the lemma. All vertices in $R^D$ are saturated by our assumption that there are no augmenting paths of length at most $2D + 1$.  So for every $d \in [0, L-1]$, we have $|R^{d}| \leq |L^{d+1}|$ as all vertices in $R^d$ are matched by vertices in $L^{d+1}$.
	
	Combining the two statements gives us $\alpha|L^{d}| \leq |L^{d+1}|$ for every $d \in [0, L-1]$. Thus $|L^D| \geq \alpha^L |L^0|$, which contradicts the definition of $D$ and that  $|L^0| \geq 1, |L^D| \leq |L|$. 
\end{proof}

\section{Online Rounding of Fractional Solutions for Load Balancing: $(2+O(\epsilon))$-Approximation}
\label{sec:2-approx-rounding}

In this section, we describe the rounding algorithm that converts the online fractional solutions $\bfx^{(1)}, \bfx^{(2)}, \cdots, \bfx^{(n)}$ into integral ones, with a $(2+O(\epsilon))$-approximation ratio. This shall finish the proof of Theorem~\ref{thm:unrelated-constant-approx}.

To do this, we create an online bipartite matching instance $H = (L \uplus R, E_H)$ with vertex updates on the right side, as well as a fractional matching $\bfy$ in $H$. $L$ corresponds to the jobs, and thus vertices in $L$ arrive one by one. Due to the recourse on the fractional solutions, we may also insert/delete vertices from $R$.  Then we use the algorithm in Section~\ref{subsec:b-matching-changing} to maintain an integral matching in $H$ covering $L$, which naturally leads to an assignment of jobs to machines.  In the analysis, we show that the assignments have small makespan, and the number of times we reassign jobs is bounded. 

For most part in this section, we fix a job $i \in M$,  and describe how to maintain the part of $H$ and the fractional matching $\bfy$ correspondent to the machine $i$. Many notations in the section depend on $i$, but for simplicity we do not include $i$ in these notations. 

So, we fix $i \in M$ and a time step $t$. Let $\bfx = \bfx^{(t)}$ be the fractional solution at $t$. The part of $H$ and $\bfy$ for $i$ can be derived from a 2-level partition of an interval for the jobs. We sort all the jobs $j \in N(i)$ in non-increasing order of $p_{ij}$ values: Let $N(i) = \{j_1, j_2, \cdots, j_{|N(i)|}\}$ with $p_{ij_1} \geq p_{ij_2} \geq \cdots \geq p_{ij_{|N(i)|}}$.  We then map each job $j_o$ to a \emph{job interval} $I_o$ of length $x_{ij_o}$: $j_1$ is mapped to $I_1 := [0, x_{ij_1}]$,  $j_2$ is mapped to $I_2 := [x_{ij_1}, x_{ij_1} + x_{ij_2}]$, $j_3$ is mapped to $I_3: = [x_{ij_1} + x_{ij_2}, x_{ij_1} + x_{ij_2} + x_{ij_3}]$, and so on.  Let $X = \sum_{j \in N(i)} x_{ij}$ be the fractional number of jobs assigned to $i$; that is, the total length of all job intervals. When $\bfx$ gets updated, $X$ and the job intervals also change accordingly. 

Then we define the 2-level partition of the interval $[0, X]$.  In the first level, we partition $[0, X]$ into a set of intervals (we distinguish between a normal interval and a job interval), that we call \emph{buckets}.  We guarantee that all buckets have length between $\frac1\epsilon$ and $\frac4\epsilon$. The only exception is that when $X < \frac{1}{\epsilon}$, in which case the whole $[0, X]$ is just one bucket.  In the second level, each bucket is partitioned into intervals, that we call \emph{segments}.  All the segments except the last one in a bucket have length between $1-3\epsilon$ and $1-\epsilon$. The last bucket has length at most $1-\epsilon$. 

We now describe how to maintain the 2-level partition of the interval $[0, X]$, as the fractional solution $\bfx$ changes. Initially, $X = 0$ and there is one bucket and segment of $0$-length at $0$. We show how to handle two operations: increasing some $x_{ij_o}$ value, and decreasing some $x_{ij_o}$ value. The two operations are sufficient for our algorithm: when a new job $j$ with $ij \in E$ arrives, we create a $0$-length job interval for the job at the appropriate position, and then we increase $x_{ij}$ from $0$ to the desired number. The changes to $\bfx$ can also be handled using the two operations.

\newcommand{\seg}{{\mathsf{seg}}}
\newcommand{\buc}{{\mathsf{buc}}}

First, consider the case where we need to increase $x_{ij_o}$ for some $o$. We find the first segment $\seg$ (from left to right) that internally intersects the job interval $I_o$.  We then increase $x_{ij_o}$ and the lengths  of $I_o$ and $\seg$ continuously at the same rate.  As a result, the length of the bucket $\buc$ containing $\seg$ also increases. Also, the segments after $\seg$ and the buckets after $\buc$ will be shifted to the right continuously. We run the procedure until  $x_{ij_o}$ is increased enough, or one of the following events happens.
\begin{itemize}
	\item The length of $\seg$ reaches $1-\epsilon$. In this case, we re-divide the bucket $\buc$ into segments: all the segments except the last segment in $\buc$ have length exactly $1-2\epsilon$, and the last bucket have length at most $1-2\epsilon$. 
	\item The length of $\buc$ reaches $\frac4\epsilon$. In this case, we divide $\buc$ into two buckets of length $\frac2\epsilon$ each. Then we re-divide each of the two buckets into segments of length $1 - 2\epsilon$ as in the previous case. 
\end{itemize}
If we have not increased $x_{ij_o}$ enough after handling the event, we continue running the procedure.

Now suppose we need to decrease $x_{ij_o}$ by some amount. We find the first segment $\seg$ that intersects the job interval $I_o$ internally. Let $\buc$ be the bucket containing $\seg$. Similarly, we decrease $x_{ij_o}$ and the lengths  of $I_o$, $\seg$ and $\buc$ continuously at the same rate, until we decreased $x_{ij_o}$ enough, or $\seg$ does not intersect internally with $I_o$ anymore, or one of the following two events happen.  As before, the segments after $\seg$ and the buckets after $\buc$ will be shifted to the left continuously.
\begin{itemize}
	\item The length of $\seg$ drops to $1-3\epsilon$ and $\seg$ is not the last segment of $\buc$. In this case, again we re-divide $\buc$ into segments of length $1-2\epsilon$ as before. 
	\item The length of $\buc$ drops to $\frac1\epsilon$ and $\buc$ is not the only bucket.  In this case, we merge $\buc$ with either its previous or the next bucket, depending on which one exists. If the merged bucket has length at most $\frac3\epsilon$, we then keep it. Otherwise, we divide the bucket into 2 equal-length buckets, each of which has length between $\frac{1.5}\epsilon$ and $\frac{2.5}{\epsilon}$. In any case, we divide each of the newly created buckets into segments of length $1-2\epsilon$. 
\end{itemize}
Again, if we have not decreased $x_{ij_o}$ enough after handling the event, we continue the decreasing operation. 

The 2-level partition of $X$ for $i$ determines the portion of the graph $H$ and the fractional matching vector $\bfy$.  For every segment $\seg$, we have a vertex $v_\seg \in R$ for the segment. For every job $j_o \in N(i)$ such that $\seg \cap I_o\neq \emptyset$, we have $j_ov_\seg \in E_H$, and $y_{j_ov_\seg}$ is the length of the intersection between $\seg$ and the job interval $I_o$.  We highlight some technicality here: All the intervals we defined are closed interval. So it is possible that the intersection of $\seg$ and $I_o$ is only one point. In this case we have $j_ov_\seg \in E_H$ and $y_{j_ov_\seg} = 0$. It is allowed for the bipartite matching algorithm to match $j_o$ to $v_\seg$. 

\begin{claim}
    \label{claim:2-approx-y-equal-x}
	At any time $t$, for every $ij \in E$, we have $\sum_{v \text{ is segment for } i} y_{jv} = x_{ij}$.
\end{claim}
\begin{proof}
	This holds since all the segments form a partition of $[0, X]$, and the job interval for $j$ is inside $[0, X]$. 
\end{proof}
Therefore, if we consider the whole graph $H$, $\bfy$ defines a fractional matching between $L = [t]$ and $R$ at time $t$: every $j \in L$ is covered to an extent of $1$ and every $v \in R$ is matched by an extent of at most $1-\epsilon$. This implies the following claim:

\begin{claim}
    \label{claim:lb-H-expand}
	At any time,  we have $|N_H(A)| \geq \frac{|A|}{1-\epsilon}$ for every $A \subseteq L$. 
\end{claim}

We then finish the description of the rounding algorithm. We maintain the dynamically-changing bipartite graph $H$, and use the bipartite-matching algorithm described in~\Cref{subsec:b-matching-changing} and Theorem~\ref{thm:online-matching-changing} to maintain a bipartite matching. If a job $j$ is assigned to a segment for machine $i$, we then assign job $j$ to machine $i$ in the load-balancing instance. 

\subsection{Analysis}
Now we analyze the recourse and the competitive ratio of the algorithm respectively. 
\paragraph{Analysis of Recourse} First we analyze the recourse made by the rounding algorithm. 
\begin{lemma}
    \label{lemma:lb-number-vertex-updates}
	For any $t \in [n]$, the number of vertex updates on segments for $i$ by any time $t$ is at most $O\left(\frac{1}{\epsilon^2}\right) \cdot \sum_{t' = 1}^t \sum_{j \in N(i)} |x^{(t')}_{ij} - x^{(t'-1)}_{ij}|$.
\end{lemma}
\begin{proof}
	We fix the machine $i$ in this proof, and consider two different causes for vertex updates for machine $i$.
	
	First, vertex updates may happen when a new bucket is created. When a new bucket is created, its length is between $\frac{1.5}{\epsilon}$ and $\frac{3}{\epsilon}$: When a bucket is created due to an increasing operation, then the new buckets have length $2/\epsilon$.  Consider a decreasing operation. The merged bucket will have length between $\frac2\epsilon$ and $\frac5\epsilon$. If it has length at most $3\epsilon$, then no splitting happens and its length is between $\frac2\epsilon$ and $\frac3\epsilon$. Otherwise, the two new buckets will have length between $\frac{1.5}{\epsilon}$ and $\frac{2.5}{\epsilon}$.  So, it takes $0.5/\epsilon$ fractional recourse on $(x_{ij})_j$ to create a new bucket. Moreover, when a new bucket is created, $O(1/\epsilon)$ segments will be inserted and deleted. Therefore, creation of new buckets incur at most $O(1) \cdot \sum_{t' = 1}^t \sum_{j \in N(i)} |x^{(t')}_{ij} - x^{(t'-1)}_{ij}|$.
	
	Then we consider vertex updates due to re-division of a bucket into segments. Whenever a re-division happens, all segments have lengths exactly $1 - 2\epsilon$ except for the last one, which has length at most $1-2\epsilon$. The next re-division happens if the length of some segment increases to $1-\epsilon$, or the length of some segment that is not the last one decreases to $1-3\epsilon$. So, it takes $\epsilon$ fractional recourse on $(x_{ij})_j$ for the next re-division to happen. When a re-division happens, we make $O(1/\epsilon)$ vertex updates.  Therefore, re-divisions incur at most $O\big(\frac1{\epsilon^2}\big) \cdot \sum_{t' = 1}^t \sum_{j \in N(i)} |x^{(t')}_{ij} - x^{(t'-1)}_{ij}|$ vertex updates.
	
	Notice that these are the only case where a vertex update happens. We need to mention that when the length of  the intersection of $\seg$ and $I_o$ decreases to $0$,  no vertex update happens, since still there is an edge between $j_o$ and $v_\seg$ in $H$. 
\end{proof}

Now we can complete the analysis of the recourse of the online rounding algorithm. Note that the total recourse made by the load balancing algorithm is at most the number of reassignments made by the bipartite-matching algorithm. By Theorem~\ref{thm:online-matching-changing} and Claim~\ref{claim:lb-H-expand}, the recourse of the latter by time $t$ is at most $O\big(\frac{\log n}{\epsilon}\big) \big(t + \#(\text{vertex updates by }t)\big)$. By Lemma~\ref{lemma:lb-number-vertex-updates} and summing up the bounds over all machines $i$, the number of vertex updates by time $t$ is at most $O\big(\frac {1}{\epsilon^2}\big)\sum_{t' =1}^{t} \left| \bfx^{(t')} - \bfx^{(t'-1)} \right|$, which is at most $O\big(\frac{t}{\epsilon^3}\big)$ by Lemma~\ref{lemma:lb-maintain-fractional}.  Therefore the total recourse by time $t$ made by the load balancing algorithm is at most $O\left(\frac{t\log n}{\epsilon^4}\right)$. Recall that the extra $O\left(\frac{\log(1/\epsilon)}{\epsilon}\right)$-factor comes from making the assumption that $T^*$ is known. 

\paragraph{Analysis of Competitive Ratio} It remains for us to analyze the competitive ratio of the algorithm. Again we fix a machine $i$ and prove the following:
\begin{lemma}
	The load of machine $i$ at any time is at most $(2+O(\epsilon))T^*$.
\end{lemma}
\begin{proof}
		Fix a time $t$ and $\bfx = \bfx^{(t)}$.  We now use $v_{ik}$ to denote the $k$-th segment for $i$ from left to right, over all the buckets.  Let $J_{ik}$ be the set of jobs in $N(i)$ whose job intervals internally intersect the segment $v_{ik}$ at the time.  We use $\pi_{jik}$ to indicate if $j$ is assigned to $v_{ik}$ in the solution for bipartite matching instance at time $t$.  
		
		Notice that each all segments except for last segments of buckets have length between $1-3\epsilon$ and $1-\epsilon$.  The last segment of a bucket has length at most $1-\epsilon$.   All buckets have length between $1/\epsilon$ and $4/\epsilon$; the only exception is when there is only one bucket, whose length might be smaller than $1/\epsilon$.
		
		The total load $i$ receives at time $t$ is then
	\begin{align*}
		\sum_{k} \sum_{j \in J_{ik}} \pi_{jik} p_{ij} & \leq \frac{1}{1-3\epsilon}\sum_{k} \sum_{j \in J_{ik}} y_{j, v_{ik}} p_{ij} + T^* + \epsilon\cdot \sum_{k} \sum_{j \in J_{ik}} y_{j, v_{ik}}p_{ij} = (1+O(\epsilon)) \sum_{k} \sum_{j \in J_{ik}} y_{j, v_{ik}} p_{ij} + T^* \\
		& = (1+O(\epsilon)) \sum_{j \in N(i)} x_{ij} p_{ij} + T^* \leq (1+O(\epsilon)) T^* + T^* = (2+O(\epsilon)) T^*.
	\end{align*}
	We need to elaborate more on the first inequality. For every $j \in J_{ik}$ with $\pi_{jik} = 1$, we can try to upper bound $p_{ij}$ by $\frac{1}{1-3\epsilon}\sum_{j' \in J_{i(k-1)}}y_{j', v_{i(k-1)}}p_{ij'}$.  Notice that every $j'$ in the summation has $p_{ij'} \geq p_{ij}$.  If $k$ is not the first segment of any bucket, then the upper bound holds as the total $\bfy$ value for the previous segment is at least $1-3\epsilon$.  The $p_{ij}$ for the first segment of the first bucket can be upper bounded by $T^*$; the $p_{ij}$ for the first segment of other buckets can be bounded by $\epsilon$ times the budget from the previous bucket, as the bucket has length at least $1/\epsilon$. The last equality is by Claim~\ref{claim:2-approx-y-equal-x}. The last inequality comes from that the total fractional load on machine $i$ is at most $(1+\epsilon)T^*$.
\end{proof}
 	\section{Lower Bound on Recourse for Load Balancing in the Fully Dynamic Model} \label{appendix:FD-LB}
	In this section we consider the online load balancing problem with recourse in the fully dynamic model. Again, we have a set $M$ of machines and a set $J$ of $n$ jobs. Jobs can arrive and depart, and at any time, we are guaranteed that there is a schedule of makespan at most $T^*$, for a given $T^*$. To achieve an $\alpha$-competitive ratio, our algorithm needs to maintain a solution of makespan at most $\alpha T^*$.  An algorithm with an amortized recourse of $\beta$ can make at most $\beta t$ reassignments for the first $t$ arrival/departure events, for any $t$. By making copies of jobs, we assume every job arrives and departs exactly once; so there are $2n$ events in the sequence. Our negative result holds even for the restricted assignment setting, and when we only need to maintain a fractional schedule.  Indeed, the lower bound holds in the offline setting: the best algorithm that knows everything upfront must incur a large recourse.   

	\begin{theorem}
		\label{thm:lb}
		Let $\alpha(n)  = o(\log n)$ be a monotone non-decreasing function of $n$. There is an instance for the above problem such that the following holds. Any offline algorithm that maintains a fractional schedule of $\alpha(n) T^*$ needs to incur an amortized recourse of $n^{\Omega(1/\alpha(n))}$. 
	\end{theorem}

	To see why an algorithm needs to incur a large recourse in the fully dynamic model, consider the following simple instance. There are 2 machines, and $n'$ small jobs of size $1$ each arrive at the beginning of the algorithm. Each of them can be assigned to the two machines.  Big jobs have size $n'$ and they are 2 types of them: type-1 big jobs can only be assigned to machine 1, and type-2 big jobs can only be assigned to machine 2. Consider the following online updates: a type-1 big job arrives and departs, a type-2 big job arrives and departs, then a type-1 big job arrives and departs, and so on. Then at any time, there is a schedule of makespan $n'$ for active jobs. If the sequence is long enough, then any $(1.5-c)$-competitive algorithm must incur an amortized recourse of $\Omega(cn')$.  The constant 1.5 can be made arbitrarily close to 2, if we introduce more machines. However, to go beyond 2, we need to use a recursive construction, using the basic instance as a building block.  \medskip
	
	Since our algorithm only needs to maintain fractional schedules, we can make the problem more general by allowing each job $j$ to have a reassignment cost $c_j \in \Z_{\geq 0}$: Reassigning $x$ fraction of a job $j$ incurs a recourse of $xc_j$.  To make the reassignment costs uniform without changing the instance, one can break a job of size $p_j$ with reassignment cost $c_j$ into $c_j$ jobs of size $p/c_j$, each with reassignment cost $1$.  This will change the number of jobs in the instance, and we take care of the issue at the end of this section.  We compare the total recourse of the algorithm against $\sum_{j \in J} c_j$.  (Recall that every job arrives and departs only once).
	
	Now describe the instance with reassignment costs. Let $L \geq 1, P \geq 2^L$ be two integers. We construct a perfect binary tree $\bfT$ of with $L$ levels of edges; so there are $2^L$ leaves in the tree. The level of a vertex is $L$ minus its distance to the root in the tree: The root has level $L$ and the laves have level $0$.  Let $\bfr, \bfV, \bfV^\circ$ be the root, vertex set and leaf set of the tree $\bfT$ respectively.  For every $v \in \bfV \setminus \bfV^\circ$, let $\rmleft(v)$ and $\rmright(v)$ be the left and right child of $v$ respectively. For every $v \in \bfV$, let $\ell(v)$ be the level of $v$ and $\Lambda^\circ(v)$ be the set of leaves that are descendants of $v$.
	
	There are 2 machines at each leaf and so there are in total $m : = 2^{L+1}$ jobs.  There is a job $j_v$ of size $1$ for every $v \in \bfV$, and the job $j_v$ can be assigned to any machine in $\Lambda^\circ(v)$.  The reassignment cost $c_{j_v}$ of the job $j_v$  is defined as
	\begin{align*}
		c_{j_v} := P^{\ell(v)}.
	\end{align*} To guarantee that every job arrives and departs once, we make copies of these jobs so that each copy is used only once. 
		
	\begin{algorithm}
		\caption{construct-instance$(v)$ \hfill //$v \in \bfV$}
		\label{alg:construct-instance}
		\begin{algorithmic}[1]
			\State let $2^{\ell(v)}$ copies of the job $j_v$ arrive
			\If{$\ell > 0$}
			\State{\textbf{repeat} $P$ times:}
			\State \hspace{\algorithmicindent} construct-instance$(\ell-1, \rmleft(v))$
			\State \hspace{\algorithmicindent} construct-instance$(\ell-1, \rmright(v))$
			\EndIf
			\State let the $2^\ell$ copies of the job $j_v$ created at step $1$ depart
		\end{algorithmic}
	\end{algorithm}

	The instance is constructed recursively by calling the procedure $\text{construct-instance}(\bfr)$, defined in Algorithm~\ref{alg:construct-instance}. After Step 1 of the procedure construct-instance$(v)$ for some leaf $u$, we have $2^L + 2^{L-1} + 2^{L-2} + \cdots +  2^0 = 2^{L+1} - 1$ active jobs:  For any ancestor $v$ of $u$, we have $2^{\ell(v)}$ active copies of $j_v$. It is easy to see that the active jobs can be scheduled on the machines so that each machine takes at most 1 job: For every strict ancestor $v$ of $u$ at level $\ell$, we schedule the $2^\ell$ copies of $j_v$ on $\Lambda^\circ(v')$, where $v'$ is the child of $v$ that is not an ancestor of $u$; the 1 copy of $j_u$ is scheduled on one machine at $u$ (the other machine is idle). Therefore, the optimum makespan at any time is at most $1$. 
	
	We show that any algorithm that maintains a fractional solution of makespan at most $\frac{L+1}{3}$ need to incur a large recourse.  First, we make our task easier, and we lower bound the recourse needed for this easier task.  In the task, the algorithm is allowed to \emph{remove} fractions of jobs, where the recourse for removing $x$ fraction of a job $j$ is $xc_j$. This means that a job does not need to be fully scheduled. The only requirement needed is the following: At the beginning of each iteration of Loop 3 in an execution of construct-instance$(v)$, the $2^{\ell(v)}$ copies of $j_v$ created must be brought back to machines (at $\Lambda^\circ(v)$). At the time, the algorithm is allowed to distribute the $2^\ell$ jobs fractionally on the machines at $\Lambda^\circ(v)$ in any manner, without incurring any recourse. 
	
	Clearly, the new task is not harder as any algorithm for the original task is an algorithm for the new one with the same recourse.  For the new task, the following assumptions on the algorithm can be made.
	\begin{itemize}
		\item First, we only need the removing operations, not the moving operations.  We can change the operation of moving a faction of a job from machine $i$ to $i'$ to the operation of removing the fractional job from $i$. This does not change the recourse of the operation, and can only decrease the loads of machines.  Also, bringing back jobs back later does not incur any recourse. 
		\item Then, focus on one iteration of Loop 3 in an execution of construct-instance$(v)$. We can assume all the removing operations for the $2^{\ell(v)}$ copies of $j_v$ in this operation are done immediately after they were brought back to the schedule.  Earlier removal can only decrease the machine loads. As a result, we can assume the following: at the beginning of the iteration, 
		 some fraction of the $2^{\ell(v)}$ jobs are put on the machines, and the recourse incurred is $c_{j_v}$ times the fractional number of jobs that are not put on the machines.  No removing is allowed anymore during the iteration. 
		\item 	Finally, for an execution of construct-instance$(v)$, we can assume the algorithm uses the same procedure for the $P$ iterations of Loop 3, as there are identical. 
	\end{itemize}

	Therefore, the algorithm for the modified task can be specified by one fractional assignment $f_v: \Lambda^\circ(v) \to \R_{\geq 0}$ for every vertex $v \in \bfV$: $f_v(u)$ the fractional number of copies of $j_v$ we assign to the two machines on $u$ in an execution of construct-instance$(v)$. It is guaranteed that $\sum_{u \in \Lambda^\circ(v)} f_v(u) \leq 2^{\ell(v)}$.
	
	Suppose at any time, the fractional schedule has makespan at most $\frac{L+1}{3}$. Then we have
	\begin{align}
		\sum_{v\text{ ancestor of }u} f_v(u) \leq 2 \times \frac{L+1}{3} = \frac{2(L+1)}{3}, \forall u \in \bfV^\circ.
		\label{inequ:lb-cong}
	\end{align}
	Summing up \eqref{inequ:lb-cong} over all leaves $u$ gives
	\begin{align*}
		\sum_{v \in \bfV, u \in \Lambda^\circ(v)} f_v(u) \leq \frac{2(L+1)}{3}\cdot 2^L.
	\end{align*}
	
	The recourse of the algorithm is then
	\begin{align*}
		&\quad \sum_{v \in \bfV}P^{L - \ell(v) + 1} \cdot c_{j_v} \cdot \left(2^{\ell(v)} - \sum_{u \in \Lambda^\circ(v)}f_v(u)\right)  =  P^{L + 1} \sum_{v \in \bfV}  \left(2^{\ell(v)} - \sum_{u \in \Lambda^\circ(v)}f_v(u)\right) \\
		&\geq P^{L+1}\left((L+1) \cdot 2^L- \frac{2(L+1)}{3}\cdot 2^L\right) = \frac{P^{L+1}\cdot(L+1)\cdot 2^L}{3}.
	\end{align*}
	In the above sequence, $P^{L - \ell(v) + 1}$ is the number of iterations of Loop 3 in all executions of construct-instance$(v)$, $c_{j_v} = P^{\ell(v)}$ is the per-unit recourse for copies of job $j_v$, and $\left(2^{\ell(v)} - \sum_{u \in \Lambda^\circ(v)}f_v(u)\right)$ is the fractional number of copies of $j_v$ that we did not schedule in each iteration of Loop 3. 

	The sum of $c_j$ over all arrived jobs is
	\begin{align*}
		\sum_{v \in \bfV} \big(2^{\ell(v)} \cdot P^{L - \ell(v)}\big) \cdot P^{\ell(v)} = P^L\sum_{v \in \bfV} 2^{\ell(v)} = P^L \cdot (L+1) \cdot 2^L.
	\end{align*}
	Above, $2^{\ell(v)}$ is the number of copies we create in each execution of construct-instance$(v)$, $P^{L-\ell(v)}$ is the number of times we run construct-instance$(v)$, and $P^{\ell(v)} = c_{j_v}$ is the recourse of a copy of $j_v$.

	Therefore, we proved that the total recourse is at least $\Omega(P)$ times $\sum_j c_j$.  Splitting jobs to make all reassignment costs equaling 1, we obtain an instance with $n = (L+1)\cdot (2P)^L$ jobs, with the amortized recourse $\Omega(P)$.  Therefore, for any integer $\alpha$, we can set $L = 3 \alpha - 1$, and $P = \floor{\left({\frac{n}{L+1}}\right)^{1/L}/2} = n^{\Omega(1/\alpha)}$.  This finishes the proof of Theorem~\ref{thm:lb}.

\end{document}